\documentclass[english]{article}

\usepackage[utf8]{inputenc}

\usepackage{geometry}
\usepackage{enumitem}
\setenumerate[1]{leftmargin=3em}

\usepackage{amssymb}
\usepackage{amsmath}
\usepackage{amsthm}
\usepackage{amsopn}

\usepackage{mathrsfs}
\usepackage{mdframed}
\usepackage{tikz-cd}

\usepackage{esint}

\usepackage{thmtools, thm-restate}
\usepackage{hyperref}

\usepackage{verbatim}
\usepackage{mathtools}
\usepackage{fullpage}
\usepackage{dsfont}

\usepackage{physics}
\usepackage{setspace}
\usepackage{epstopdf}
\usepackage[T1]{fontenc}
\usepackage{graphicx}

\usepackage{scalerel,stackengine}

\usepackage{ulem}

\usepackage{authblk}

\usepackage{soul}

\theoremstyle{remark}
\newtheorem{Remark}{Remark}[section]

\theoremstyle{plain}
\newtheorem{Theorem}{Theorem}[section]
\newtheorem{Lemma}{Lemma}[section]
\newtheorem{Corollary}{Corollary}[section]
\newtheorem{Definition}{Definition}[section]
\newtheorem{Proposition}{Proposition}[section]

\newtheorem{assump}{Assumption}

\DeclareRobustCommand{\rchi}{{\mathpalette\irchi\relax}}
\newcommand{\irchi}[2]{\raisebox{\depth}{$#1\chi$}}

\def \qed {\hfill \vrule height6pt width 6pt depth 0pt}

\newcommand{\R}{\mathbb{R}}
\newcommand{\N}{\mathbb{N}}

\newcommand{\p}{\partial}

\usepackage{ulem}

\allowdisplaybreaks

\numberwithin{equation}{section}

\stackMath
\newcommand\reallywidehat[1]{%
	\savestack{\tmpbox}{\stretchto{%
			\scaleto{%
				\scalerel*[\widthof{\ensuremath{#1}}]{\kern-.6pt\bigwedge\kern-.6pt}%
				{\rule[-\textheight/2]{1ex}{\textheight}}
			}{\textheight}%
		}{0.5ex}}%
	\stackon[1pt]{#1}{\tmpbox}%
}

\renewcommand{\tilde}{\widetilde}
\renewcommand{\leq}{\leqslant}
\renewcommand{\geq}{\geqslant}

\title{Combined mean-field and semiclassical limits\\ of large fermionic systems\footnote{Chen, L., Lee, J. \& Liew, M. Combined Mean-Field and Semiclassical Limits of Large Fermionic Systems. J Stat Phys 182, 24 (2021). https://doi.org/10.1007/s10955-021-02700-w}}
\author[1]{Li Chen\thanks{chen@math.uni-mannheim.de}}
\author[2]{Jinyeop Lee\thanks{jinyeoplee@kias.re.kr}}
\author[1]{Matthew Liew\thanks{mliew@mail.uni-mannheim.de}}
\affil[1]{Institut f\"ur Mathematik, Universit\"at Mannheim}
\affil[2]{School of Mathematics, Korea Institute for Advanced Study}
\date{}

\begin{document}
	\setcounter{page}{1}
	\maketitle
	
	\begin{abstract}
		We study {the} time dependent Schr\"odinger equation for large spinless fermions with the semiclassical scale $\hbar = N^{-1/3}$ in three dimensions. By using the Husimi measure defined by coherent states, we rewrite the Schr\"odinger equation into a BBGKY type of hierarchy for the $k$ particle Husimi measure. Further estimates are derived to obtain the weak compactness of the Husimi measure, and in addition uniform estimates for the remainder terms in the hierarchy are derived in order to show that in the semiclassical regime the weak limit of the Husimi measure is exactly the solution of the Vlasov equation. 
	\end{abstract}

\section{Introduction}
{In this paper, we aim to study the combined mean-field and semiclassical limit of $N$-fermions from time-dependent Schr\"odinger equation to Vlasov equation. 
	The following anti-symmetric subspace of $L^2(\R^{3N})$ is considered for fermions,
	\begin{equation*}
	L^2_a(\R^{3N}) := \left\{ \Psi \in L^2(\R^{3N}) : \Psi (q_{\pi(1)}, \dots, q_{\pi(N)}) = \varepsilon(\pi)  \Psi (q_1, \dots, q_N) \right\}.
	\end{equation*}
	It is known that a system of fermions initially confined in a volume of order one have kinetic energy of order $N^{5/3}$ due to the Pauli principle. Therefore, to balance the order, the scale of the interaction term should be of order $N^{-1/3}$, 	we refer to \cite{benedikter2014mean,Benedikter2016book} for more details about this scaling. After a time rescaling of $N^{1/3}$ the Sch\"odinger equation for $N$-fermions is written into	\begin{equation*}
	N^{\frac{1}{3}}\mathrm{i} \partial_t \Psi_{N,t} =\left[ - \frac{1}{2} \sum_{j=1}^N  \Delta_{q_j} + \frac{1}{2N^\frac{1}{3}} \sum_{i \neq j}^N V(q_i - q_j)\right] \Psi_{N,t}.
	\end{equation*}
	By denoting the semiclassical scale $\hbar = N^{-1/3}$ and multiplying both sides by $\hbar^2$, one can recover the $N^{-1}$, the coupling constant for the mean field interaction. Hence one arrives at {the} following many body Schr\"odinger equation
	\begin{equation}\label{Schrodinger_1}
	\begin{cases}
	\displaystyle \mathrm{i} \hbar \partial_t \Psi_{N,t} = \left[ - \frac{\hbar^2}{2} \sum_{j=1}^N  \Delta_{q_j} + \frac{1}{2N} \sum_{i \neq j}^N V(q_i - q_j) \right]  \Psi_{N,t}=:H_N\Psi_{N,t}, \\
	\Psi_{N,0} = \Psi_N,
	\end{cases}
	\end{equation}
	where $\Psi_{N,t} \in L^2_a(\R^{3N})$, $\Psi_N$ is the initial data in $L_a^2(\R^{3N})$, and $V$ is the interacting potential.
	
	The limit from many body Schrödinger equation to the Vlasov equation has been studied extensively in the literature. Narnhofer and Sewell \cite{Narnhofer1981} and Spohn \cite{Spohn1981} are the first to prove this limit with the potential $V$ assumed to be analytic and $C^2$ respectively.
	
	For large $N$, in the mean field limit regime, the solution of many body fermionic Schr\"odinger equation can be approximated by the solution of the following nonlinear Hartree-Fock equation,
	\begin{equation*}
	\begin{cases}
	\mathrm{i} \hbar \p_t \omega_{N,t} = \left[ -\hbar^2 \Delta + (V*\varrho_t) - X_t, \omega_{N,t} \right], \\
	\omega_{N,0} = \omega_N,
	\end{cases}
	\end{equation*}
	where $\omega_{N,t}$ is the one-particle density matrix, $\varrho_t(q) = N^{-1} \omega_{N,t} (q;q)$ and $X_{N,t}$ is a small term having the kernel $X_{t}(x,y) = N^{-1} V(x-y)  \omega_{N,t} (x;y)$.
	In \cite{ELGART20041241}, { for the initial data being a Slater determinant, the approximation} has been proved for short time for analytic interaction potential by using BBGKY hierarchy, while \cite{benedikter2014mean} proved {the approximation with convergence rate} for arbitrary time and weakened potential in the framework of second quantization. Similar results have been extended for mixed states in \cite{benediktermixed} and for relativistic case in \cite{benedikter2014rel}.
	Recently, with the help of Fefferman-de la Llave decomposition \cite{Fefferman1986,hainzl2002general}, weaker assumptions on the interaction potential have been considered. Specifically, Coulomb potential has been considered in \cite{Porta2017}, inverse power law in \cite{saffirio2017mean}. Further relevant literature on the fermionic case for the {mean-field} limit problem of Schr\"odinger equation can be found in \cite{BACH20161,Frohlich2011,petrat2014derivation,Petrat2017,Petrat2016ANM}.

	In parallel, the mean field limit for the bosonic case from many body Schr\"odinger system to nonlinear Hartree equation was proved in \cite{erdos2001derivation} for Coulomb potential. Also for Coulomb potential, the convergence with rate $N^{1/2}$ has been obtained in \cite{rodnianski2009quantum}. 
	Later, it has been optimized to the optimal convergence rate $N^{-1}$ in \cite{Chen2011}, and furthermore for stronger singular potentials in \cite{Chen2018}.

	The semiclassical limit from Hartree-Fock equation to Vlasov equation has been obtained in the literature by using Wigner-Weyl transformation of the one-particle density matrix $\omega_{N,t}$ defined by 
	\begin{equation}\label{Wigner}
	W_{N,t} (q,p) = \left( \frac{\hbar}{2\pi} \right)^3 \int \mathrm{d}y\ e^{-\mathrm{i}p\cdot y} \omega_{N,t} \left( x+\frac{\hbar}{2}y; x-\frac{\hbar}{2}y \right),
	\end{equation}
	which has been intensively studied in the semiclassical limit of quantum mechanics by Lions and Paul in \cite{Lions1993}. 
	In \cite{benedikter2016Hartree} the authors compared the inverse Wigner transform of the Vlasov solution and the solution of Hartree-Fock and get the convergence rate in the trace norm as well as Hilbert-Schmidt norm with the regular assumptions on the initial data. The works in this direction have also been extended for inverse power law potential \cite{Saffirio2019}, convergence rate in Schatten norm in \cite{lafleche2020strong}, and Coulomb potential and mixed states in \cite{saffirio2019Hartree}. The convergence of relativistic Hartree dynamic to relativistic Vlasov equation has also been considered in \cite{Dietler2018}. Further convergence results from Hartree to Vlasov can be found in \cite{amour2013classical,amour2013,Gasser1998,Markowich1993}.
	
	It is known that Wigner transform \eqref{Wigner} is not a true probability density as it may be negative in certain phase-space. In fact, \cite{Hudson1974,PhysRevA.79.062302,doi:10.1063/1.525607} concludes that the Wigner measure is non-negative if and only if the pure quantum states are Gaussian, whilst  \cite{doi:10.1063/1.531326} state that the Wigner measure is non-negative if the state is a convex combination of coherent states. Nevertheless, it has been shown that if one convolutes the Wigner measure with a Gaussian function in phase-space, it will yield a non-negative probability measure known as Husimi measure \cite{Fournais2018,Combescure2012,Zhang2008}. In fact, from \cite[p.21]{Fournais2018}, the Husimi measure is given by
	\begin{equation}\label{gaussian_m_vs_W}
	m_{N,t}^{(k)}  = \frac{N(N-1)\cdots (N-k+1)}{N^k} W^{(k)}_{N,t} * \mathcal{G}^\hbar,
	\end{equation}
	where $1 \leq k \leq N$, $\mathcal{G}^\hbar = (\pi \hbar)^{-3k} \exp \big(-\hbar^{-1} (\sum_{j=1}^k |q_j|^2 + |p_j|^2) \big)$ and $W^{(k)}_{N,t}$ is the Wigner transform of $k$-particle density matrix. 
	
	In the recent development, the convergence to Vlasov equation in the semiclassical Wasserstein pseudo-distance has been proved in \cite{Golse2017,Golse2019,golse:hal-01334365,Lafleche2019GlobalSL,Lafleche2019PropagationOM}.  The semiclassical Wasserstein pseudo-distance is computed between the Husimi measure and Vlasov solution.
	
	One can also show the combined limit by first taking the semiclassical limit and then the mean field limit from many particle Schr\"odinger to Vlasov via the Liouville equations, and the corresponding BBGKY hierarchy\footnote{See Figure \ref{nice_figure}}. This has been done in \cite{Golse2017}.
}	

\begin{figure}
	\begin{center}
		\begin{tikzcd}[column sep=3cm,row sep=3cm,cells={nodes={draw=black}}]
		\text{$N$-fermionic Schr\"odinger} \arrow[r, "N \to \infty"] \arrow[d,"\hbar \to 0"] \arrow[rd,"\hbar=N^{-1/3}\to 0"{sloped,auto,pos=0.5},red]
		& \text{Hartree Fock} \arrow[d, "\hbar \to 0"] \\
		\text{Liouville}\arrow[r, "N \to \infty"]& \text{Vlasov}
		\end{tikzcd}
	\end{center}
	\caption{Relations of $N$-fermionic Schrödinger systems to other mean-field equations \cite{Golse2016,Golse2017}.}\label{nice_figure}
\end{figure}
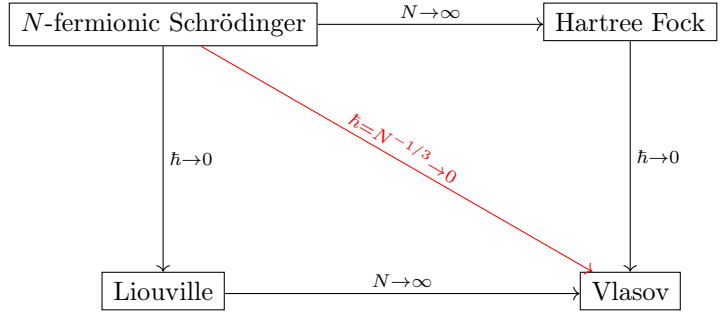

{
	Our goal, therefore, is to obtain the Vlasov equation from Schr\"odinger equation directly, as shown in the  diagonal line of Figure \ref{nice_figure}, by taking $N \to \infty$ and $\hbar \to 0$ simultaneously. In order to do this,} it is convenient for us to introduce the second quantization framework in our study { of} the quantum many-body systems. In particular, we utilize the notations in \cite{benedikter2014mean,Benedikter2016book,Chen2011} where the fermionic Fock space is defined as
\[
\mathcal{F}_a = \bigoplus_{n \geq 0} L^2_a (\R^{3n}, (\mathrm{d}x)^{\otimes n}),
\]
where we denote $(\mathrm{d}x)^{\otimes n} = \mathrm{d}x_1 \cdots \mathrm{d}x_n$. The creation and annihilation operator in terms of their respective distributive forms,
\begin{equation}\label{distrb_annil_creation_def}
a^*(f) = \int \mathrm{d}x\ a^*_x f(x), \quad a(f) = \int \mathrm{d}x\ a_x \overline{f(x)}.
\end{equation}

Due to the canonical anti-commutator relation (CAR) in the fermionic regime, we have that for all $f,g \in H^1(\R^3)$
\begin{equation}
\{ a(f), a^*(g)\} = \left< f,g\right>, \quad \{ a^*(f), a^*(g)\} = \{ a(f), a(g)\} = 0,
\end{equation}
where $\{A, B\} = AB + BA$ is the anti-commutator. In particular, the CAR for operator kernels hold as follow
\begin{equation}\label{eq:CAR}
\{ a_x, a^*_y\} = \delta_{x=y}, \quad \{a^*_x, a^*_y\} = \{ a_x, a_y\} = 0.
\end{equation}
This CAR in distributive form will be frequently used in our computations.

As in \cite{benedikter2014mean}, we may write the corresponding Hamiltonian in terms of the operator valued distribution in $\mathcal{F}_a$ by
\begin{equation}\label{Fock_Hamil}
\mathcal{H}_N = \frac{\hbar^2}{2} \int \mathrm{d}x\ \nabla_x a^*_x \nabla_x a_x + \frac{1}{2N} \iint \mathrm{d}x\mathrm{d}y\ V(x-y) a^*_x a^*_y a_y a_x.
\end{equation}
Therefore, we rewrite the Schr\"odinger equation in Fock space as follows,
\begin{equation}\label{Schrodinger_fock}
\begin{cases}
\displaystyle \mathrm{i} \hbar \partial_t \psi_{N,t} =\mathcal{H}_N \psi_{N,t}, \\
\psi_{N,0} = \psi_N,
\end{cases}
\end{equation}
for all $\psi_{N,t} \in \mathcal{F}_a^{(N)} $ and $t \in [0,T]$, where $\psi_N\in \mathcal{F}^{(N)}_a$ with $\|\psi_N\|=1$. The solution to the above Cauchy problem is $\psi_{N,t} = e^{-\frac{i}{\hbar} \mathcal{H} t} \psi_{N}$, with a given initial data $\psi_{N}$.

\begin{Remark}
	It should be noted the states $\psi_{N,t}$ in our analysis stays in the $N$th-sector of $\mathcal{F}_a$ due to the definition of Husimi measure which will be given later. Therefore, denoting $\mathcal{F}_a^{(n)}$ to be the $n$-th sector in $\mathcal{F}_a$, we say that $\psi_{N,t} \in \mathcal{F}_a^{(N)}$ for all $t \geq 0$.
\end{Remark}

Furthermore, we use the definition of the number and kinetic energy operators as follows,
\begin{equation}\label{def_NumOp_KE}
\mathcal{N} =  \int \mathrm{d}x\ a^*_x a_x \quad \text{and}\quad \mathcal{K} =  \hbar^2 \int \mathrm{d}x\ \nabla_x a^*_x \nabla_x a_x,
\end{equation}
respectively. We further explore the properties of the operators in \eqref{def_NumOp_KE} in section \ref{subset_Numop}. 

Next, we shall introduce the Husimi measure. In fact, our notation follows closely with the notations in Fournais, Lewin and Solovej  \cite{Fournais2018} where it deals with large fermionic particles in stationary case. The main tool in their analysis is the use of coherent state, a subtle tool that proves extremely useful in our work as well.

For any real-valued normalized function $f$, the coherent state is given by,\footnote{The function $f$ can be any real-valued function.\cite{Fournais2018} For this paper, we set $f$ to be compactly supported. See Assumption \ref{assume_f_compact}.}
\begin{equation}\label{coherent_def}
f^{\hbar}_{q, p} (y) := \hbar^{-\frac{3}{4}} f \left(\frac{y-q}{\sqrt{\hbar}} \right) e^{\frac{\mathrm{i}}{\hbar} p \cdot y },
\end{equation}
Similar to \cite{Antonio2016} and \cite{Fournais2018}, the $k$-particle Husimi measure is defined as, for any $1 \leq k \leq N$
\begin{equation}\label{husimi_def}
m^{(k)}_{N} (q_1, p_1, \dots, q_k, p_k) := \left< \psi_{N},  a^*(f^{\hbar}_{q_1, p_1})\cdots  a^*(f^{\hbar}_{q_k, p_k} )a(f^{\hbar}_{q_k, p_k})\cdots  a(f^{\hbar}_{q_1, p_1} ) \psi_{N} \right>,
\end{equation}
where $\psi_{N} \in \mathcal{F}_a^{(N)}$ is the $N$-fermionic states, $a(f^\hbar_{q,p})$ and $a^*(f^\hbar_{q,p})$ are the annihilation and creation operators respectively. Husimi measure defined in \eqref{husimi_def} measures how many particles, in particularly fermions, are in the $k$ semiclassical boxes with length scaled of $\sqrt{\hbar}$ centered in its respectively phase-space pair, $(q_1, p_1), \dots, (q_k, p_k)$. 

In the context of this paper, we use $m_{N,t}^{(k)}$ to be the time dependent Husimi measure defined by the solution of the Schr\"odinger equation $\psi_{N,t}$. By using operator kernels defined in \eqref{distrb_annil_creation_def},  we may rewrite the Husimi measure as follows
\begin{equation}\label{husimi_def_1}
\begin{aligned}
&m^{(k)}_{N,t} (q_1, p_1, \dots, q_k, p_k)\\
&:= \dotsint (\mathrm{d}w\mathrm{d}u)^{\otimes k} \left( f^\hbar_{q,p}(w) \overline{f^\hbar_{q,p}(u)} \right)^{\otimes k}  \left< \psi_{N,t},  a^*_{w_1} \cdots a^*_{w_k} a_{u_k} \cdots a_{u_1} \psi_{N,t} \right>,
\end{aligned}
\end{equation}
where the tensor products indicate
\[
(\mathrm{d}w\mathrm{d}u)^{\otimes k} := \mathrm{d}w_1 \mathrm{d}u_1 \cdots \mathrm{d}w_k \mathrm{d}u_k
\]
and
\[
\quad \left( f^\hbar_{q,p}(w) \overline{f^\hbar_{q,p}(u)} \right)^{\otimes k} := \prod_{j=1}^k f^\hbar_{q_j,p_j}(w_j) \overline{f^\hbar_{q_j,p_j}(u_j)}.
\]

Note that the function $f$ here is a very well localized function in practice \cite{Fournais2018}, therefore we may take the following assumption
\begin{assump}\label{assume_f_compact}
	The real-valued function $f \in L^2\cap W^{1,\infty}(\R^3)$ satisfies $\norm{f}_2 = 1$, and has compact support.
\end{assump}

Additionally, we assume that the interaction potential to satisfy

\begin{assump}\label{assume_V}
	$V$ is a real-valued function such that $V(-x) = V(x)$ and $V\in W^{2,\infty}(\R^3)$.
\end{assump}

As is well known that in the mean field semiclassical regime, the dynamic of \eqref{Schrodinger_1} can be approximated by a one particle Vlasov equation. Namely, for all $q, p \in \R^3$
\begin{equation}\label{classical_Vlasov}
\partial_t m_t(q, p)  + p \cdot \nabla_{q} m_t(q, p)  = \nabla \big( V * \rho_t\big)(q) \cdot \nabla_{p}  m_t(q, p),
\end{equation}
with initial data $m_0(q,p)$,
where $m_t(q,p)$ is the time dependent one particle probability density function, and  $\rho_t (q) = \int m_t(q, p)  \mathrm{d}p$.
Although \eqref{classical_Vlasov} is a non-linear equation, such equation would be more suitable to analyze than the increasingly large systems of Schr\"odinger equation. The well-posedness of the above Vlasov problem is given by Drobrushin \cite{Dobrushin1979} for smooth $V$. 

Now, we are ready to state the our main results.

\begin{Theorem} \label{thmmain}	Let Assumptions \ref{assume_f_compact} and $\ref{assume_V}$ hold, $\psi_{N,t}$ be the solution of Schr\"odinger equation \eqref{Schrodinger_fock}, $m^{(k)}_{N,t}$ be the Husimi measure defined in \eqref{husimi_def_1}. If 
	$m^{(1)}_N$, the $1$-particle Husimi measure of the initial data $\psi_N$, { satisfies}
	\begin{equation}\label{momentinitial}
	\iint \mathrm{d}q_1\mathrm{d}p_1 (|{p}_1|^2+|{q}_1|)m^{(1)}_{N}(q_1,p_1)\leq C.
	\end{equation}
	Then, for all $t\geq 0$, the $k$-particle Husimi measure at time $t$, $m^{(k)}_{N,t}$ has a weakly convergent subsequence which converges to $m_{t}^{(k)}$ in $L^1(\R^6)$, where $m_{t}^{(k)}$ is a weak solution of the following infinite hierarchy in the sense of distribution, i.e. it satisfies for all $k\geq 1$ that
	\begin{align}\label{infinitehierarchy}
	&\p_t m_{t}^{(k)}(q_1,p_1,\dots,q_k,p_k)+ \vec{p}_k \cdot \nabla_{\vec{q}_k}m_{t}^{(k)}(q_1,p_1,\dots,q_k,p_k)  \\
	&= \frac{1}{(2\pi)^3}  \nabla_{\vec{p}_k} \cdot \iint \mathrm{d}q_{k+1}\mathrm{d}p_{k+1} \nabla V(q_j - q_{k+1})m_{t}^{(k+1)}(q_1,p_1,\dots,q_{k+1},p_{k+1}).\nonumber
	\end{align}
\end{Theorem}

{ By using \cite[Theorem 7.12]{Villani2003}, we have the following corollary,}

\begin{Corollary}\label{cor:1-Wasserstein}	Suppose assumptions \ref{assume_f_compact} and $\ref{assume_V}$ hold. Assume further that the initial data of \eqref{infinitehierarchy} can be factorized, i.e. for all $k\geq 1$, 
	\begin{equation}
	\label{IC}
	\|m^{(k)}_{N}-m_0^{\otimes k}\|_{L^1}\rightarrow 0, \quad \mbox{ as } N\rightarrow\infty.
	\end{equation}
	Then, if the infinite hierarchy \eqref{infinitehierarchy} has a unique solution and $m_t$ { is} the solution to the classical Vlasov equation in \eqref{classical_Vlasov}, it holds that
	\[
	W_1 \left(m^{(1)}_{N,t}\; , \; m_t \right) \longrightarrow 0, \quad \mbox{ as } N \to \infty,
	\]
	for $t \geq 0$.
\end{Corollary}

{
	\begin{Remark}
		In the pioneering work by Spohn \cite{Spohn1981}, 
		he considered
		\begin{eqnarray*}
			&&r_n^{(N)}(\xi_1,\eta_1,\dots,\xi_N,\eta_N,t) \\
			&=& \operatorname{tr}\left[ e^{-\mathrm{i}H_Nt}|\Psi_N\rangle\langle\Psi_N| e^{\mathrm{i}H_Nt} \prod_{j=1}^N \exp{(\mathrm{i} (N^{-1/3}\xi_j p_j + \eta_j x_j))} \right]
		\end{eqnarray*}
		with $p_j= -\mathrm{i} \nabla_j$ and obtained the following Vlasov hierarchy,
		\begin{align*}
		&\frac{\partial}{\partial t} r_n^{(N)} (\xi_1,\eta_1,\dots,\xi_n,\eta_n,t) = \sum_{j=1}^{n} \eta_j \frac{\partial}{\partial \xi_j} r_n^{(N)} (\xi_1,\eta_1,\dots,\xi_n,\eta_n,t)\\
		&+\sum_{j=1}^n \int \hat{V}(\mathrm{d}k) k\cdot \xi_j r_{n+1}^{(N)} (\xi_1,\eta_1,\dots,\xi_j,\eta_j+k,\dots\xi_n,\eta_n,0,-k,t),
		\end{align*}
		which is slightly different from Vlasov hierarchy for Husimi measure given in \eqref{infinitehierarchy}, or the version in \eqref{BBGKY_k} before taking the limit.
		The benefit of the hierarchy in  \eqref{BBGKY_k} is that one observes directly the mean field and semiclassical structure in the remainder terms. The explicit formulation is helpful in getting estimates for the remainder terms in \eqref{BBGKY_k}. Moreover if one can handle singular potentials (or even the Coulomb potential) for both terms separately, one expects that this new approach can be applied to obtain the limit from many body Schr\"ordinger to Vlasov with singular potentials in the future. Since the mean field limit with singular potential has been studied with convergence rate, for example in \cite{Benedikter2016book}, then we can utilize similar ideas to handle one of the remainder term which includes the mean field structure. In parallel, we can apply the techniques in semiclassical limit, for example in \cite{Saffirio2019}, to get estimates for the other remainder term. 
	\end{Remark}
}

\begin{Remark}
	{ Although the results in this article does not yield a convergent rate}, the main purpose of this article is to present an alternative approach and framework, namely to rewrite the Schr\"odinger equation into a BBGKY type of hierarchy, and to derive estimates for the remainder terms that appear in the new hierarchy. 
\end{Remark}

\begin{Remark}
	In Corollary \ref{cor:1-Wasserstein}, the convergence is stated in terms of $1$-Wasserstein distance. For completeness, we give its definition as defined in \cite{Villani2003}
	\begin{equation}
	W_1 (\mu, \nu) := \max_{\pi \in \Pi (\mu,\nu)} \int |x-y|\ \mathrm{d}\pi(x,y),
	\end{equation}
	where $\mu$ and $\nu$ are probability measures and $\Pi(\mu,\nu)$ the set of all probability measures with marginals $\mu$ and $\nu$. The Wasserstein distance, also known as Monge-Kantorovich distance, is a distance 
	{on the set of probability measures.}
	In fact, if we interpret the metric in $L^p$ space as the distance that measures two densities ``vertically'', the Wasserstein distance measures the distance between two densities ``horizontally''\cite{Santambrogio2015}.
\end{Remark}	

\begin{Remark}
	The assumptions for initial data \eqref{momentinitial} and \eqref{IC} can be realized by choosing $\psi_{N}$ to be the Slater-determinant. That is, for all orthonormal basis $\{ \varphi_j \}_{j=1}^\infty$, the initial data is given as
	\begin{equation}\label{initial_is_slater}
	\psi_{N}(q_1,\dots,q_N) = \frac{1}{\sqrt{N!}} \mathrm{det} \{ \varphi_j(q_i)  \}_{1 \leq i,j \leq N},
	\end{equation}
\end{Remark}

\begin{Remark} Assumptions  \ref{assume_f_compact} and $\ref{assume_V}$ are expected to be weakened to the situation that $f\in H^1(\R^3)$, $|x|f(x)\in L^2(\R^3)$, and $V$ to be Coulomb potential. These will be our future projects.
\end{Remark} 

\begin{Remark}
	In this context, we have applied the BBGKY hierarchy, the intermediate mean field approximation Hartree Fock system has not been benefited. With Hartree Fock approximation, one can do direct factorization in the equation for $m^{(1)}_{N,t}$. In this direction, we expect to derive the rate of convergence in an appropriate distance between the Husimi measure and the solution of the Vlasov equation.
\end{Remark}

The arrangement of the paper is the following. In section 2, we give the main strategy of the proof. Followed by the reformulation of Schr\"odinger equation into a hierarchy of the Husimi measure, a sequence of necessary estimates on number operators, the localized number operators, and the kinetic energy operator are given, which will be contributed to do compactness argument for the Husimi measure. We leave the computation of the hierarchy to section \ref{proof_of_reformulation}. Furthermore, the uniform estimates for remainder terms in the hierarchy, which is another main contribution of this article, are provided in section \ref{proof_of_estimations}.

\section{Proof strategy through BBGKY type hierarchy for Husimi measure}
We first start from the many particle Schr\"odinger equation and derive an approximated hierarchy of time dependent Husimi measure by direct computation. 
Compare to the BBGKY hierarchy of Liouville equation in the classical sense, it has two families of remainder terms, which are determined by the $N$ particle wave function from Schr\"odinger equation. 
In order to take a convergent subsequence of the $k$-particle Husimi measure, we derive the uniform estimates for number operator and the kinetic energy. 
Together with an additional estimate for localized number operator, we can show that the remainder terms are of order $\hbar^{\frac{1}{2}-\delta}$, for arbitrary small $\delta$. Then the desired result will be obtained by the uniqueness of solution to the infinite hierarchy.

\subsection{Reformulation: Hierarchy of time dependent Husimi measure}\label{sec2.1}
In this subsection, we begin by examining the dynamics of $k$-particle Husimi measure by using the $N$-body fermionic Schr\"odinger. The proofs of the following propositions are provided in section \ref{proof_of_reformulation}.

\begin{Proposition} \label{lemma_vla_k1}
	Suppose $\psi_{N,t} \in \mathcal{F}_a^{(N)}$ is anti-symmetric $N$-particle state satisfying the Schr\"odinger equation in \eqref{Schrodinger_fock}. Moreover, if $V(-x)=V(x)$ then we have the following equation for $k=1$,
	\begin{equation} \label{BBGKY_k1}
	\begin{aligned}
	&\p_t m^{(1)}_{N,t}(q_1,p_1) + p_1 \cdot \nabla_{q_1} m^{(1)}_{N,t}(q_1,p_1)\\
	&= \frac{1}{(2\pi)^3}\nabla_{p_1} \cdot  \iint \mathrm{d}q_2\mathrm{d}p_2 \nabla V(q_1-q_2) m^{(2)}_{N,t}(q_1,p_1,q_2,p_2)
	+ \nabla_{q_1}\cdot \mathcal{R}_1 +\nabla_{p_1}\cdot \widetilde{\mathcal{R}}_1,
	\end{aligned}
	\end{equation}
	where the remainder terms $\mathcal{R}_1$ and $\widetilde{\mathcal{R}}_1$, are given by
	\begin{equation}\label{bbgky_remainder_1}
	\begin{aligned}
	\mathcal{R}_1 :=  &  \hbar \Im \left< \nabla_{q_1} a (f^\hbar_{q_1,p_1}) \psi_{N,t}, a (f^\hbar_{q_1,p_1}) \psi_{N,t} \right>,\\
	\widetilde{\mathcal{R}}_1  := &\frac{1}{(2\pi)^3} \cdot \Re \iint \mathrm{d}w\mathrm{d}u \iint \mathrm{d}y\mathrm{d}v \iint \mathrm{d}q_2 \mathrm{d}p_2 \int_0^1 \mathrm{d}s\\
	&\hspace{1cm}\nabla V\big(su+(1-s)w - y \big) f_{q_1,p_1}^\hbar (w) \overline{f_{q_1,p_1}^\hbar (u)} f_{q_2,p_2}^\hbar (y) \overline{f_{q_2,p_2}^\hbar (v)}  \left< a_y a_w \psi_{N,t}, a_v a_u \psi_{N,t} \right>\\
	& - \frac{1}{(2\pi)^3}\iint \mathrm{d}q_2\mathrm{d}p_2 \nabla V(q_1-q_2) m^{(2)}_{N,t}(q_1,p_1,q_2,p_2),
	\end{aligned}
	\end{equation}
\end{Proposition}

\begin{Proposition} \label{lemma_vla_bbgky_hierarchy}
	For every $1 \leq i,j \leq k$ and $q_j,p_j \in \R^3$, denote $\vec{q}_k = (q_1,\dots,q_k)$ and  $\vec{p}_k = (p_1,\dots,p_k)$. Under the assumption in Proposition \ref{lemma_vla_k1}, then for $1 < k \leq N$, we have the following hierarchy
	\begin{equation} \label{BBGKY_k}
	\begin{aligned}
	&\p_t m_{N,t}^{(k)}(q_1,p_1,\dots,q_k,p_k)+ \vec{p}_k \cdot \nabla_{\vec{q}_k}m_{N,t}^{(k)}(q_1,p_1,\dots,q_k,p_k)  \\
	&= \frac{1}{(2\pi)^3} \nabla_{\vec{p}_k} \cdot \iint \mathrm{d}q_{k+1}\mathrm{d}p_{k+1} \nabla V(q_j - q_{k+1})m_{N,t}^{(k+1)}(q_1,p_1,\dots,q_{k+1},p_{k+1})\\
	&\hspace{1cm} + \nabla_{\vec{q}_k} \cdot \mathcal{R}_k +  \nabla_{\vec{p}_k} \cdot \widetilde{\mathcal{R}}_k + \widehat{\mathcal{R}}_k,\\
	\end{aligned}
	\end{equation}
	where the remainder terms are denoted as
	\begin{equation}\label{bbgky_remainder_k}
	\begin{aligned}
	\mathcal{R}_k :=&  \hbar \Im   \left< \nabla_{\vec{q}_k}\big( a (f^\hbar_{q_k,p_k}) \cdots a (f^\hbar_{q_1,p_1})\big) \psi_{N,t},  a (f^\hbar_{q_k,p_k}) \cdots a (f^\hbar_{q_1,p_1}) \psi_{N,t} \right> ,\\
	(\widetilde{\mathcal{R}}_k)_j :=&  \frac{1}{(2\pi)^{3}} \Re \dotsint (\mathrm{d}w\mathrm{d}u)^{\otimes k} \int \mathrm{d}y \left[ \int_0^1 \mathrm{d}s \nabla V(su_j + (1-s)w_j -y) \right] \left( f_{q,p}^\hbar (w) \overline{f_{q,p}^\hbar (u)} \right)^{\otimes k} \\
	& \iint \mathrm{d}\widetilde{q} d \widetilde{p}\ f^\hbar_{\widetilde{q}, \widetilde{p}} (y)\int \mathrm{d}v\ \overline{f^\hbar_{\widetilde{q}, \widetilde{p}}(v)}  \left< a_{w_k} \cdots a_{w_1} a_y \psi_{N,t}, a_{u_k} \cdots a_{u_1} a_v \psi_{N,t} \right>\\
	& - \frac{1}{(2\pi)^3}  \iint \mathrm{d}q_{k+1}\mathrm{d}p_{k+1} \nabla V(q_j - q_{k+1})m_{N,t}^{(k+1)}(q_1,p_1,\dots,q_{k+1},p_{k+1}),\\
	\widehat{\mathcal{R}}_k := 	& \frac{\hbar^2}{2}  \Im\dotsint (\mathrm{d}w\mathrm{d}u)^{\otimes k} \sum_{j\neq i}^k  \bigg[ V(u_j-u_i) - V(w_j -w_i) \bigg]  \left( f_{q,p}^\hbar (w) \overline{f_{q,p}^\hbar (u)} \right)^{\otimes k}\\
	& \left< a_{w_k} \cdots a_{w_1} \psi_{N,t}, a_{u_k} \cdots a_{u_1} \psi_{N,t} \right>\\
	\end{aligned}
	\end{equation}
\end{Proposition}

\subsection{\textit{A priori} estimates} 

In the next steps, we derive estimates in order to have compactness of each $k$-particle Husimi measure, as well as to prove that the remainder terms converge to zero in the sense of distribution. The estimates are derived directly from the solutions of the $N$-fermionic Schr\"odinger equation.

\subsubsection{Properties of coherent states and Husimi measure}
Here we give the properties of coherent states and Husimi measure provided in \cite{Fournais2018}, which will be frequently needed in our computation. Firstly, we observe that the coherent state has a projection property, that is

\begin{Lemma}[Projection of the coherent state, \cite{Fournais2018}]\label{coherent_projections} For every real-valued function $f$ satisfying $\norm{f}_2 = 1$ and the coherent states $f^\hbar_{q,p}$ defined as in \eqref{coherent_def}, we have that
	\begin{equation}\label{projection_f}
	\frac{1}{(2\pi \hbar)^3} \iint \mathrm{d}q \mathrm{d}p \ket{f^{\hbar}_{q, p}}\bra{f^{\hbar}_{q, p} } = \frac{1}{(2\pi \hbar)^3} \iint \mathrm{d}q \mathrm{d}p \left< f^{\hbar}_{q, p}, \cdot \right> f^{\hbar}_{q, p} (y) = \mathds{1}.
	\end{equation}
\end{Lemma}
Secondly, the properties of the $k$-particle Husimi measure $m^{(k)}_{N}$ is given as follows
\begin{Lemma}[Properties of $k$-particle Husimi measure, \cite{Fournais2018}]\label{prop_kHusimi}
	Suppose for $\psi_{N} \in \mathcal{F}_a^{(N)}$ is normalized. Then, the following properties hold true for $m^{(k)}_{N}$:
	\begin{enumerate}
		\item  $m^{(k)}_{N}(q_1,p_1,\dots,q_k,p_k)$ is symmetric,
		\item  $\frac{1}{(2\pi)^{3k}} \dotsint (\mathrm{d}q\mathrm{d}p)^{\otimes k} m^{(k)}_{N}(q_1,p_1,\dots,q_k,p_k) = \frac{N(N-1)\cdots (N-k+1)}{N^k}$,
		\item $\frac{1}{(2\pi \hbar)^{3}} \iint \mathrm{d}q_k  \mathrm{d}p_k\ m^{(k)}_{N}(q_1,p_1,\dots,q_k,p_k) = (N-k+1) m^{(k-1)}_{N}(q_1,p_1,\dots,q_{k-1},p_{k-1}) $, and
		\item $ 0 \leq  m^{(k)}_{N}(q_1,p_1,\dots,q_k,p_k) \leq 1$ a.e.,
	\end{enumerate}
	where $1 \leq k \leq N$.
\end{Lemma}

\begin{Remark}\label{prop_kHusimi_t}
	Note that as $\norm{\psi_N} = \norm{\psi_{N,t}}$, Lemma \ref{prop_kHusimi} is also valid if we replaced the stationary wave-function $\psi_N$, to a time-dependent $\psi_{N,t}$, for $t \geq 0$.
	Moreover, 
	it can be obtained that
	for any fixed positive integer $1\leq k \leq N$,
	\begin{equation}\label{Linftykestimate}
	0 \leq  m^{(k)}_{N,t} \leq 1 \quad \mbox{ a.e. in } \R^{6k}. 
	\end{equation}
\end{Remark}

{
	Following \cite{Fournais2018}, we define the $\hbar$-weighted Fourier transformation as follows,
	
	\begin{Definition}[$\hbar$-weighted Fourier transform]\label{def_fourier_h}
		Let $F$ be any real-valued function in $L^2(\R^3)$. We define the $\hbar$-weighted Fourier transform of $f$ to be,
		\[
		\mathcal{F}_\hbar [f](p) := \frac{1}{(2\pi \hbar)^\frac{3}{2}} \int_{\R^3}\mathrm{d}x\ f(x) e^{- \frac{\mathrm{i}}{\hbar} p \cdot x},
		\]
		and its inverse transform by $\mathcal{F}^{-1}_\hbar$. 
	\end{Definition}
	From the Definition \ref{def_fourier_h}, we have the following identity,
	\begin{equation}\label{hbar_fourier}
	\int_{\R^3} \mathrm{d}y\ G(y) F(y) =  \int_{\R^3} \mathrm{d}y\ G(y) \frac{1}{(2\pi \hbar)^3}\iint_{\R^{3\cdot 2}} \mathrm{d}p_2\mathrm{d}v\ F(v) e^{\frac{\mathrm{i}}{\hbar}p_2\cdot(y-v)},
	\end{equation}
	for any $G,F \in L^2({\R^3})$. In other words, the Dirac-delta distribution is given by
	\begin{equation}
	\delta_y (v) =   \frac{1}{(2\pi \hbar)^3}\int_{\R^{3}}\mathrm{d}p_2\  e^{\frac{\mathrm{i}}{\hbar}p_2\cdot(y-v)}.
	\end{equation}
}

\subsubsection{Number operator and localized number operator}\label{subset_Numop}
In this part, we give the bounds of number operators and its corresponding localized version, both of which are used extensively in estimating the remainder terms in \eqref{BBGKY_k1} and \eqref{BBGKY_k}.
\begin{Lemma}\label{estimate_NumOpk}
	Let $\psi_{N,t} \in \mathcal{F}^{(N)}_a$ be the solution to Schrödinger equation in \eqref{Schrodinger_1} with initial data $\|\psi_N\|=1$, the number operator $\mathcal{N}$ defined in \eqref{def_NumOp_KE}. Then, for finite $1 \leq k \leq N$, we have
	\[
	\left< \psi_{N,t}, \frac{\mathcal{N}^k}{N^k} \psi_{N,t} \right>  = 1.
	\]
\end{Lemma}

\begin{proof}
	Since $\psi_{N,t}$ satisfies the Schrödinger equation, then for $k\geq 1$,
	\[
	i \hbar \frac{\mathrm{d}}{\mathrm{d}t} \left< \psi_{N,t}, \mathcal{N}^k \psi_{N,t} \right> =  \left< \psi_{N,t}, [\mathcal{N}^k,\mathcal{H}_N] \psi_{N,t} \right> =  k \left< \psi_{N,t},\mathcal{N}^{k-1}[\mathcal{N},\mathcal{H}_N] \psi_{N,t} \right> = 0,
	\] 
	where we used the fact that $\mathcal{H}_N$ is self-adjoint and $[\mathcal{H}_N,\mathcal{N}]=0$.
	Therefore, integrating the above equation with respect to time, gives us
	\[
	\left< \psi_{N,t}, \frac{\mathcal{N}^k}{N^k} \psi_{N,t} \right>  = \left< \psi_{N}, \frac{\mathcal{N}^k}{N^k} \psi_{N} \right>= 1,
	\]
	for any $1 \leq k \leq N$.
\end{proof}

{
	\begin{Remark}
		The expectation of the number operator is the total mass of Husimi measure. In fact, observe that 
		\begin{align*}
		\left< \psi_{N,t}, \mathcal{N} \psi_{N,t} \right> = & \int \mathrm{d}x \left< \psi_{N,t}, a_x^* a_x \psi_{N,t} \right> = \int \mathrm{d}x \left< \psi_{N,t}, a_x^* \mathds{1} a_x \psi_{N,t} \right>, \\
		\intertext{Then, by \eqref{projection_f}}
		= & \frac{1}{(2\pi \hbar)^3} \iint \mathrm{d}q\mathrm{d}p \int \mathrm{d}x \left< \psi_{N,t}, a^*_x f_{q,p}^\hbar (x) \bigg(\int \mathrm{d}y\ a_y \overline{f_{q,p}^\hbar (y)} \bigg) \psi_{N,t} \right>\\
		= & \frac{1}{(2\pi \hbar)^3} \iint \mathrm{d}q\mathrm{d}p \left< \psi_{N,t}, a^*(f_{q,p}^\hbar) a(f_{q,p}^\hbar) \psi_{N,t} \right>\\
		= & \frac{1}{(2\pi \hbar)^3} \iint \mathrm{d}q\mathrm{d}p\ m^{(1)}_{N,t} (q,p)\\
		= & N,
		\end{align*}
		where we use Lemma \ref{prop_kHusimi} in the last equality. Moverover, if we repeat the projection above for $k$-times, we get
		\begin{equation}\label{MNkL1}
		\begin{aligned}
		&\frac{1}{(2\pi)^{3k}} \dotsint (\mathrm{d}q\mathrm{d}p)^{\otimes k} m^{(k)}_{N,t}(q_1,p_1,\dots,q_k,p_k)\\
		&{ \leq }\left< \psi_{N,t}, \frac{\mathcal{N}^k}{N^k} \psi_{N,t} \right> = 1,
		\end{aligned}
		\end{equation}
		where $1\leq k \leq N$ and $t \geq 0$.
	\end{Remark}
}
More importantly, we have the following estimates for localized number operators.

\begin{Lemma}[Bound on localized number operator]\label{N_hbar} Let $\psi_N \in \mathcal{F}^{(N)}_a$ such that $\norm{\psi_N} = 1$, and $R$ be the radius of a ball such that the volume is $1$. Then, for all $1\leq k \leq N$, we have
	\begin{align*}
	&\dotsint (\mathrm{d}q\mathrm{d}x)^{\otimes k} \left< \psi_N , 
	\left(\prod_{n=1}^k \rchi_{|x_n-q_n|\leq \sqrt{\hbar}R}\right)a^*_{x_1}\cdots a^*_{x_k} a_{x_k}\cdots a_{x_1} \psi_N  \right>
	\leq \hbar^{-\frac{3}{2}k},
	\end{align*}
	where $\rchi$ is a characteristic function
\end{Lemma}

\begin{proof}
	Consider first the case where $k=1$. For every $1\leq j \leq k$, we have
	\begin{align*}
	&\int \mathrm{d}x_j \left( \int \mathrm{d}q_j\ \rchi_{|x_j-q_j|\leq \sqrt{\hbar}R} \right) \left< \psi_N ,a^*_{x_j} a_{x_j} \psi_N  \right> \\
	&= \hbar^{\frac{3}{2}}  \left< \psi_N , \mathcal{N} \psi_N  \right> = \hbar^{\frac{3}{2}-3} \left< \psi_N , \frac{\mathcal{N}}{N} \psi_N  \right>
	\leq  \hbar^{-\frac{3}{2}},
	\end{align*}
	where we used Lemma \ref{estimate_NumOpk}. Analogously, for $2 \leq k \leq N$,
	\begin{align*}
	&\int (\mathrm{d}x)^{\otimes k}   \left( \prod_{n=1}^k \int \mathrm{d}q_n\ \rchi_{|x_n-q_n|\leq \sqrt{\hbar}R} \right) \left< \psi_N , 
	a^*_{x_1}\cdots a^*_{x_k} a_{x_k}\cdots a_{x_1} \psi_N  \right>\\
	&
	\leq \hbar^{\frac{3}{2}k} \left< \psi_N , 
	\mathcal{N}^k \psi_N  \right>
	= \hbar^{\frac{3}{2}k-3k} \left< \psi_N , 
	\frac{\mathcal{N}^k}{N^k} \psi_N  \right>
	\leq \hbar^{\frac{3}{2}k-3k},
	\end{align*}
	where we applied Lemma \ref{estimate_NumOpk} again. 
\end{proof}

\begin{Lemma}[Estimate of oscillation]\label{estimate_oscillation} For $\varphi(p) \in C^\infty_0 (\R^3)$ and 
	\begin{equation}\label{estimate_oscillation_omega}
	{\Omega_\hbar^\alpha} := \{x \in \R^3;\ \max_{1\leq j \leq 3} |x_j|\leq \hbar^\alpha \},
	\end{equation}
	it holds for every $\alpha \in (0,1)$, $s \in \N$, and $x \in \R^3\backslash {\Omega_\hbar^\alpha}$, 
	\begin{equation}\label{estimate_oscillation_0}
	\left|\int_{\R^3} \mathrm{d}p\  e^{\frac{\rm i}{\hbar}p\cdot x} \varphi(p)\right| \leq C \hbar^{(1-\alpha)s},
	\end{equation}
	where $C$ depends on the compact support and the $C^s$ norm of $\varphi$.
\end{Lemma}

\begin{proof}
	We will prove the lemma in a single-variable environment. That is, we let the momentum and space to be $p=(p_1, p_2, p_3)$ and $x = (x_1, x_2, x_3)$ such that $x_j, p_j \in \R$ for all $j \in \{1,2,3\}$. Then, for arbitrary $x \in \R^3\backslash {\Omega_\hbar^\alpha}$, one of the $x_j$s is bigger than $\hbar^\alpha$. Without loss of generality, we assume that $|x_1| > \hbar^\alpha$ and $x_2, x_3 \in \R$. 
	Let ${\rm supp}$ $\varphi\subset B_r(0)\subset\R^3$, we can rewrite the left hand of \eqref{estimate_oscillation_0} into the following,
	\begin{align*}
	& \left| \int^r_{-r} \mathrm{d}p_1 \int^r_{-r} \mathrm{d}p_2\int^r_{-r} \mathrm{d}p_3  e^{\frac{\rm i}{\hbar}(p_1x_1 + p_2x_2 + p_3x_3)} \varphi(p) \right|\\
	&=  \left| \int^r_{-r} \mathrm{d}p_2  e^{\frac{\rm i}{\hbar}p_2x_2}\int^r_{-r} \mathrm{d}p_3 e^{\frac{\rm i}{\hbar}p_3x_3}\int^r_{-r} \mathrm{d}p_1  e^{\frac{\rm i}{\hbar}p_1x_1} \varphi(p) \right|
	\end{align*}
	Observe that since
	\[
	-{\rm i} \frac{\hbar}{x_1}\frac{\mathrm{d}}{\mathrm{d}p_1} e^{\frac{\rm i}{\hbar}p_1x_1} = e^{\frac{\rm i}{\hbar}p_1x_1},
	\]
	we have after $s$ times integration by parts in $p_1$,
	\begin{align*}
	&\left| \int^r_{-r} \mathrm{d}p_1 \int^r_{-r} \mathrm{d}p_2\int^r_{-r} \mathrm{d}p_3  e^{\frac{\rm i}{\hbar}(p_1x_1 + p_2x_2 + p_3x_3)} \varphi(p) \right|\\
	&=  \left|\Big(-{\rm i} \frac{\hbar}{x_1} \Big)^s \int^r_{-r} \mathrm{d}p_2  e^{\frac{\rm i}{\hbar}p_2x_2}\int^r_{-r} \mathrm{d}p_3 e^{\frac{\rm i}{\hbar}p_3x_3}\int^r_{-r} \mathrm{d}p_1  e^{\frac{\rm i}{\hbar}p_1x_1} \partial_{p_1}^s\varphi(p) \right|\\
	&
	\leq C \frac{\hbar^s}{|x_1|^s}\leq C \hbar^{(1-\alpha)s},
	\end{align*}
	where $s$ indicates the number of time that integration by parts has been performed.
\end{proof}

\subsubsection{Finite moments of Husimi measure}
To prove that the second moment in $p$ of the Husimi measure is finite, we first show that the kinetic energy is bounded from above. Recall that the definition of the kinetic energy operator $\mathcal{K}$, i.e.,
\[
\mathcal{K} = \frac{\hbar^2}{2} \int \mathrm{d}x\ \nabla_x a^*_x \nabla_x a_x,
\]
and the kinetic energy associated with $\psi_N$ is given as $ \left<\psi_N, \mathcal{K} \psi_N \right>$. 

\begin{Lemma}\label{kinetic_finite}
	Assume $ V \in W^{1,\infty}$, then the kinetic energy is bounded in the following
	\begin{equation}\label{k_kinetic_bounded}
	\left<\psi_{N,t}, \frac{\mathcal{K}}{N} \psi_{N,t}\right> \leq 2\left<\psi_{N}, \frac{\mathcal{K}}{N} \psi_{N}\right>+Ct^2,
	\end{equation}
	where $C$ depends on $\|\nabla V\|_{\infty}$.
\end{Lemma}

\begin{proof}
	From the Schrödinger equation, we get
	\begin{equation}\label{kinetic_finite_1}
	i \hbar \frac{\mathrm{d}}{\mathrm{d}t} \left<\psi_{N,t} , \mathcal{K} \psi_{N,t} \right> = \left<\psi_{N,t}, [\mathcal{K},\mathcal{H}] \psi_{N,t} \right>.
	\end{equation}
	Note that since the commutator between kinetic and interaction term is given as
	\begin{align*}
	[\mathcal{K}, \mathcal{H}] = & \frac{\hbar^2}{4}\left[\int\mathrm{d}x\ \nabla_x a^*_x \nabla_x a_x , \iint \mathrm{d}y \mathrm{d}z\ V(y-z) a^*_y a^*_z a_z a_y \right] \\
	= & \frac{\hbar^2}{4} \iint \mathrm{d}x \mathrm{d}y \nabla_x V (x-y) \bigg( \nabla_x a^*_x a^*_y a_y a_x - a^*_x a^*_y a_y \nabla_x a_x \bigg)\\
	= &  \frac{\hbar^2}{2N} \Im \iint \mathrm{d}x \mathrm{d}y \nabla_x V (x-y) (\nabla_x a^*_x a^*_y a_y a_x )
	\end{align*}
	Then, from \eqref{kinetic_finite_1}, we have that
	\[
	\frac{1}{N} \frac{\mathrm{d}}{\mathrm{d}t}  \left<\psi_{N,t} , \mathcal{K} \psi_{N,t} \right> =  \frac{\hbar}{2N^2} \Im \iint \mathrm{d}x \mathrm{d}y\ \nabla_x V (x-y) \left<\psi_{N,t}, \nabla_x a^*_x a^*_y a_y a_x  \psi_{N,t} \right>.
	\]
	Now, observe that
	\begin{align*}
	& \left| \frac{\hbar}{2N^2} \iint \mathrm{d}x \mathrm{d}y\ \nabla_x V (x-y) \left<\psi_{N,t}, \nabla_x a^*_x a^*_y a_y a_x  \psi_{N,t} \right> \right| \\
	\leq &\frac{\hbar}{2N^2} \norm{\nabla V}_{L^\infty} \iint \mathrm{d}x \mathrm{d}y\ \norm{a_y \nabla_x a_x \psi_{N,t}}\norm{a_y a_x \psi_{N,t}}\\
	\leq & C \frac{\hbar}{2N^2} \left(\iint \mathrm{d}x \mathrm{d}y \left< \psi_{N,t},\nabla_x a^*_x a^*_y   a_y \nabla_x a_x \psi_{N,t} \right> \right)^\frac{1}{2} \left(\iint \mathrm{d}x \mathrm{d}y \left< \psi_{N,t}, a^*_x a^*_y   a_y  a_x \psi_{N,t} \right> \right)^\frac{1}{2} \\
	= & C  \left( \frac{\hbar^2}{N} \int \mathrm{d}x \left< \psi_{N,t},\nabla_x a^*_x \frac{\mathcal{N}}{N} \nabla_x a_x \psi_{N,t} \right> \right)^\frac{1}{2} \left< \psi_{N,t}, \frac{\mathcal{N}^2}{N^2} \psi_{N,t} \right>^\frac{1}{2}\\
	\leq & C \left( \left< \psi_{N,t},\frac{\mathcal{K}}{N} \psi_{N,t} \right> \right)^\frac{1}{2},
	\end{align*}
	Thus, we have 
	\[
	\frac{\mathrm{d}}{\mathrm{d}t}  \left< \psi_{N,t},\frac{\mathcal{K}}{N} \psi_{N,t} \right> \leq C  \left< \psi_{N,t},\frac{\mathcal{K}}{N} \psi_{N,t} \right>^\frac{1}{2}.
	\]
	Integrating both sides with respect to time $t$ and we obtain the desired inequality.

\end{proof}

\begin{Proposition} \label{2nd_moment_finite} For $t \geq 0$, assume \ref{assume_f_compact} and let $m^{(k)}_{N,t}$ to be the $k$-particle Husimi measure. Denoting the phase-space vectors $\vec{q}_k = (q_1, \dots, q_k)$ and $\vec{p}_k = (p_1, \dots, p_k)$, we have the following finite moments,
	\[
	\dotsint (\mathrm{d}q \mathrm{d}p)^{\otimes k}\ (|\vec{q}_k| + |\vec{p}_k|^2) m^{(k)}_{N,t}(q_1, \dots, p_k) \leq C(1+t^3)
	\]
	where $C$ is a constant dependent on $k$, $\iint \mathrm{d}q_1 \mathrm{d}p_1 (|q_1| + |p_1|^2) m^{(1)}_{N}(q_1, p_1) $, and $\|\nabla V\|_\infty$.
\end{Proposition}

\begin{proof}
	We first consider the case where $k=1$. Observe that we may rewrite the kinetic energy {as follows}
	\begin{align*}
	&\frac{1}{N} \left<\psi_{N,t}, \mathcal{K} \psi_{N,t} \right>= \frac{\hbar^2}{N} \int \mathrm{d}w\  \left<\psi_{N,t}, \nabla_w a^*_w \nabla_w a_w \psi_{N,t} \right> \\
	= & \frac{\hbar^2}{N} (2\pi \hbar)^{-3} \iint \mathrm{d}q_1\mathrm{d}p_1\iint \mathrm{d}w\mathrm{d}u\ f^\hbar_{q_1,p_1} (w) \overline{f^\hbar_{q_1,p_1} (u)} \left<\psi_{N,t}, \nabla_w a^*_w \nabla_u a_u \psi_{N,t} \right>\\
	= & \frac{\hbar^2}{(2\pi)^3}  \iint \mathrm{d}q_1\mathrm{d}p_1\iint \mathrm{d}w\mathrm{d}u\ \nabla_w f^\hbar_{q_1,p_1} (w) \overline{\nabla_u f^\hbar_{q_1,p_1} (u)} \left<\psi_{N,t},  a^*_w  a_u \psi_{N,t}\right>\\
	=&   \frac{\hbar^2}{(2\pi)^3} \iint \mathrm{d}q_1\mathrm{d}p_1\iint \mathrm{d}w\mathrm{d}u\ ( - \nabla_{q_1} + i \hbar^{-1} p_1) f^\hbar_{q_1,p_1} (w)\cdot( - \nabla_{q_1} - i \hbar^{-1} p_1) \overline{ f^\hbar_{q_1,p_1} (u)} \left<\psi_{N,t},  a^*_w  a_u \psi_{N,t}\right>,
	\end{align*}
	where we used the fact that
	\[
	\nabla_w f\left(\frac{w-q_1}{\sqrt{\hbar}} \right) = - \nabla_{q_1} f\left(\frac{w-q_1}{\sqrt{\hbar}} \right).
	\]
	To continue, we have
	\begin{equation}\label{moment_momentum_0}
	\begin{aligned}
	\frac{1}{N} \left<\psi_{N,t}, \mathcal{K} \psi_{N,t} \right> = & \frac{1}{(2\pi)^3}\iint \mathrm{d}q_1\mathrm{d}p_1\ |p_1|^2  m^{(1)}_{N,t}(q_1, p_1)\\ 
	& +\frac{\hbar^2}{(2\pi)^3}  \iint \mathrm{d}q_1\mathrm{d}p_1\iint \mathrm{d}w\mathrm{d}u\ \nabla_{q_1} f^\hbar_{q_1,p_1} (w)\cdot  \nabla_{q_1}\overline{ f^\hbar_{q_1,p_1} (u)}\left<\psi_{N,t},  a^*_w  a_u \psi_{N,t}\right>\\
	& + \hbar \frac{2i}{(2\pi)^3} \Im  \iint \mathrm{d}q_1\mathrm{d}p_1\iint \mathrm{d}w\mathrm{d}u\ p_1 \cdot \nabla_{q_1}f^\hbar_{q_1,p_1} (w) \overline{ f^\hbar_{q_1,p_1} (u)}\left<\psi_{N,t},  a^*_w  a_u \psi_{N,t}\right>.\\
	\end{aligned}
	\end{equation}
	Since kinetic energy is real-valued, if we take the real part of \eqref{moment_momentum_0}, the last term in the right hand side vanishes since it is purely imaginary, yielding
	\begin{align*}
	\frac{1}{N} \left<\psi_{N,t}, \mathcal{K} \psi_{N,t} \right> = & \frac{1}{(2\pi)^3}\iint \mathrm{d}q_1\mathrm{d}p_1\ |p_1|^2  m^{(1)}_{N,t}(q_1, p_1)\\ 
	& +\frac{\hbar^2}{(2\pi)^3} \Re \iint \mathrm{d}q_1\mathrm{d}p_1\iint \mathrm{d}w\mathrm{d}u\ \nabla_{q_1} f^\hbar_{q_1,p_1} (w)\cdot  \nabla_{q_1}\overline{ f^\hbar_{q_1,p_1} (u)}\left<\psi_{N,t},  a^*_w  a_u \psi_{N,t}\right>.
	\end{align*}
	Note that by \eqref{hbar_fourier}, we have
	\begin{align}
	& \nonumber\frac{\hbar^2}{(2\pi)^3} \iint \mathrm{d}q_1\mathrm{d}p_1\iint \mathrm{d}w\mathrm{d}u\  \nabla_{q_1}f^\hbar_{q_1,p_1} (w)\cdot \nabla_{q_1}  \overline{ f^\hbar_{q_1,p_1} (u)}\left<\psi_{N,t},  a^*_w  a_u \psi_{N,t}\right> \\
	\label{moment_momentum_2}=& \hbar^{2+3} \iint \mathrm{d}q_1\mathrm{d}w\ \hbar^{-\frac{3}{2}} \left| \nabla_{q_1} f\left(\frac{w-q_1}{\sqrt{\hbar}}\right) \right|^2 \left<\psi_{N,t},  a^*_w  a_w \psi_{N,t}\right> \\
	\nonumber = & \hbar \int \mathrm{d}\widetilde{q}  \left| \nabla f\left(\widetilde{q}\right) \right|^2 \left<\psi_{N,t},  \frac{\mathcal{N}}{N} \psi_{N,t}\right>\\
	\nonumber= & \hbar \int \mathrm{d}\widetilde{q}  \left| \nabla f\left(\widetilde{q}\right) \right|^2,
	\end{align}
	where we recall that $\hbar^3 = N^{-1}$. Thus, taking the real part of \eqref{moment_momentum_0}, we have that
	\begin{equation}
	\begin{aligned}
	\left<\psi_{N,t},\frac{\mathcal{K}}{N}\psi_{N,t} \right> =  \frac{1}{(2\pi)^3}\iint \mathrm{d}q_1\mathrm{d}p_1\ |p_1|^2  m^{(1)}_{N,t}(q_1, p_1)
	+ \hbar \int \mathrm{d}q \left| \nabla f\left(q\right) \right|^2,
	\end{aligned}
	\end{equation}
	which means, 
	\begin{equation}\label{moment_momentum_1}
	\begin{aligned}
	\frac{1}{(2\pi)^3}\iint \mathrm{d}q_1\mathrm{d}p_1\ |p_1|^2  m^{(1)}_{N,t}(q_1, p_1) 
	\leq  \left<\psi_{N,t},\frac{\mathcal{K}}{N} \psi_{N,t} \right> .
	\end{aligned}
	\end{equation}
	
	Therefore, \eqref{moment_momentum_1} tells us that the second moment of the $1$-particle Husimi measure in momentum space is finite if the kinetic energy is finite. 
	
	Now, we turn our focus on the moment with respect to position space. From \eqref{BBGKY_k1}, we get
	\begin{align*}
	&\p_t \iint \mathrm{d}q_1 \mathrm{d}p_1\ |q_1| m^{(1)}_{N,t}(q_1,p_1) =  \iint |q_1|\p_tm^{(1)}_{N,t}(q_1,p_1)\\
	= &\iint \mathrm{d}q_1 \mathrm{d}p_1\ |q_1| \bigg( - p_1 \cdot \nabla_{q_1} m^{(1)}_{N,t}(q_1,p_1) + \frac{1}{(2\pi)^3} \nabla_{p_1} \cdot \iint \mathrm{d}w\mathrm{d}u \iint \mathrm{d}x\mathrm{d}y \iint \mathrm{d}q_2 \mathrm{d}p_2 \int_0^1 \mathrm{d}s\\
	&  \nabla V\big(su+(1-s)w - x \big) f_{q_1,p_1}^\hbar (w) \overline{f_{q_1,p_1}^\hbar (u)} f_{q_2,p_2}^\hbar (x) \overline{f_{q_2,p_2}^\hbar (y)}  \left< a_x a_w \psi_{N,t}, a_y a_u \psi_{N,t} \right> + \nabla_{q_1} \cdot \mathcal{R}_1 \bigg).
	\intertext{Then, using intergration by parts with respect to $p_1$,}
	&=  \iint \mathrm{d}q_1 \mathrm{d}p_1\ \nabla_{q_1}|q_1| \cdot  \left( p_1 m^{(1)}_{N,t}(q_1,p_1) +  \mathcal{R}_1 \right)\\
	&= \iint \mathrm{d}q_1 \mathrm{d}p_1\ \frac{q_1}{|q_1|} \cdot  \left( p_1 m^{(1)}_{N,t}(q_1,p_1) +  \mathcal{R}_1 \right)\\
	& \leq \iint \mathrm{d}q_1 \mathrm{d}p_1  \left( |p_1| m^{(1)}_{N,t}(q_1,p_1) +  |\mathcal{R}_1| \right),
	\end{align*}
	where $R_1$ is the remainder term in \eqref{bbgky_remainder_1}.
	
	Note that by Young's product inequality, we have
	\begin{align*}
	\iint \mathrm{d}q_1 \mathrm{d}p_1 \ | p_1| m^{(1)}_{N,t}(q_1,p_1) \leq & \iint \mathrm{d}q_1 \mathrm{d}p_1  \left( 1+ |p_1|^2 \right) m^{(1)}_{N,t}(q_1,p_1) \\
	&\leq (2\pi)^3  \left(1 + 2\left<\psi_{N},\frac{\mathcal{K}}{N} \psi_{N} \right> + Ct^2\right),
	\end{align*}
	where we used \eqref{moment_momentum_1} and Lemma \ref{kinetic_finite} in the last inequality. Next, we want to bound the term associated with $\mathcal{R}_1$,
	\begin{align*}
	\iint \mathrm{d}q_1 \mathrm{d}p_1 \  |\mathcal{R}_1| \leq & \hbar  \iint \mathrm{d}q_1 \mathrm{d}p_1 \  | \left< \nabla_{q_1} a (f^\hbar_{q_1,p_1}) \psi_{N,t}, a (f^\hbar_{q_1,p_1}) \psi_{N,t} \right>|.
	\end{align*}
	
	Observer that we have,
	\begin{align*}
	&  \hbar  \iint \mathrm{d}q_1 \mathrm{d}p_1 \  \Big|\left< \nabla_{q_1} a (f^\hbar_{q_1,p_1}) \psi_{N,t},  a (f^\hbar_{q_1,p_1}) \psi_{N,t} \right>  \Big| \leq \hbar  \iint \mathrm{d}q_1 \mathrm{d}p_1 \norm{\nabla_{q_1} a (f^\hbar_{q_1,p_1}) \psi_{N,t},}\norm{ a (f^\hbar_{q_1,p_1}) \psi_{N,t} }\\
	\leq & \hbar  \left[\iint \mathrm{d}q_1 \mathrm{d}p_1  \left< \nabla_{q_1} a (f^\hbar_{q_1,p_1}) \psi_{N,t},   \nabla_{q_1} a (f^\hbar_{q_1,p_1}) \psi_{N,t} \right> \right]^\frac{1}{2} \left[\iint \mathrm{d}q_1 \mathrm{d}p_1\ \left<\psi_{N,t}, a^*(f^\hbar_{q_1,p_1})  a (f^\hbar_{q_1,p_1}) \psi_{N,t}\ \right> \right]^\frac{1}{2}\\
	=& \left[\hbar^2 \iint \mathrm{d}q_1\mathrm{d}p_1\iint \mathrm{d}w\mathrm{d}u\  \nabla_{q_1}f^\hbar_{q_1,p_1} (w)\cdot \nabla_{q_1}  \overline{ f^\hbar_{q_1,p_1} (u)}\left<\psi_{N,t},  a^*_w  a_u \psi_{N,t}\right> \right]^\frac{1}{2} (2\pi)^\frac{3}{2}\\
	\leq & (2\pi)^3  \sqrt{\hbar} \left[  \int \mathrm{d}\widetilde{q}  \left| \nabla f\left(\widetilde{q}\right) \right|^2\right]^\frac{1}{2},
	\end{align*}
	where we used \eqref{moment_momentum_2}, Lemma \ref{prop_kHusimi}. Thus, we have that
	\begin{equation}
	\p_t \iint \mathrm{d}q_1 \mathrm{d}p_1\ |q_1| m^{(1)}_{N,t}(q_1,p_1) \leq (2\pi)^3  \left(1 + 2\left<\psi_{N},\frac{\mathcal{K}}{N} \psi_{N} \right> + Ct^2 + C \sqrt{\hbar} \right)\leq C(1+t^2).
	\end{equation}
	which gives the estimate for first moment after integrating with respect to time $t$.
	
	We now consider the case of $2 \leq k \leq N$. In this computation, we make use of the properties of $k$-particle Husimi measure. Namely, that the $m^{(k)}_{N,t}$ is symmetric and satisfies the following equation
	
	\begin{equation}\label{moment_momentum_3}
	\begin{aligned}
	\frac{1}{(2\pi)^3} \iint \mathrm{d}q_k\mathrm{d}p_k\  m^{(k)}_{N,t} (q_1,p_1,\dots,q_k,p_k)  =& \frac{(N-k+1)}{N}  m^{(k-1)}_{N,t} (q_1,p_1,\dots,q_{k-1},p_{k-1})\\
	\leq & m^{(k-1)}_{N,t} (q_1,p_1,\dots,q_{k-1},p_{k-1}).
	\end{aligned}
	\end{equation}
	Observe that for fixed $1\leq k \leq N$.
	\begin{align*}
	&\dotsint ( \mathrm{d}q\mathrm{d}p)^{\otimes k}  \sum_{j=1}^k |p_j|^2 m^{(k)}_{N,t} (q_1,p_1,\dots,q_k,p_k)\\
	= & \sum_{j=1}^k \iint  \mathrm{d}q_j\mathrm{d}p_j\ |p_j|^2 \dotsint  \mathrm{d}q_1\mathrm{d}p_1 \cdots \widehat{\mathrm{d}q_j}\widehat{\mathrm{d}p_j} \cdots \mathrm{d}q_k \mathrm{d}p_k\ m^{(k)}_{N,t} (q_1,p_1,\dots,q_k,p_k).
	\intertext{Then, by using the symmetricity of $m^{(k)}_{N,t}$ and change of variables, we get }
	&
	= k \iint  \mathrm{d}q\mathrm{d}p\ |p|^2 \dotsint  ( \mathrm{d}q\mathrm{d}p)^{\otimes k-1}  m^{(k)}_{N,t} (q,p, q_1,p_1\dots,q_{k-1},p_{k-1})\\
	&
	= (2\pi)^{3(k-1)}k \frac{(N-1)\cdots(N-k+1)}{N^{k-1}} \iint \mathrm{d}q\mathrm{d}p\ |p|^2 m^{(1)}_{N,t} (q,p)\\
	&
	\leq (2\pi)^{3k} k \left(1 + 2\left<\psi_{N},\frac{\mathcal{K}}{N} \psi_{N} \right> + Ct^2\right)\leq C(1+t^2),
	\end{align*}
	where we denoted  $(\mathrm{d}q\mathrm{d}p)^{\otimes k-1} = \mathrm{d}q_1\mathrm{d}p_1\cdots \mathrm{d}q_{k-1}\mathrm{d}p_{k-1}$.
	
	Similar strategy is used to obtain the first moment with respect to $\vec{q}_k$. That is
	\begin{align*}
	&\dotsint ( \mathrm{d}q\mathrm{d}p)^{\otimes k}  \sum_{j=1}^k |q_j| m^{(k)}_{N,t} (q_1,p_1,\dots,q_k,p_k)\\
	= &  (2\pi)^{3(k-1)}k \frac{(N-1)\cdots(N-k+1)}{N^{k-1}} \iint \mathrm{d}q\mathrm{d}p\ |q| m^{(1)}_{N,t} (q,p)\\
	\leq & (2\pi)^{3(k-1)}k  \iint \mathrm{d}q\mathrm{d}p\ |q| m^{(1)}_{N,t} (q,p)\leq  C(1+t^3).
	\end{align*}
	This yields the desired conclusion.
\end{proof}

\subsection{Uniform estimates for the remainder terms}
\label{sec2.3}
In this subsection, we give uniform estimates for the error terms that appear in \eqref{BBGKY_k1} and \eqref{BBGKY_k}. They are all bounded of order $\hbar^{\frac{1}{2}-\delta}$ for arbitrary small $\delta>0$. The proofs of all the following propositions will be provided in section \ref{proof_of_estimations}.

\begin{Proposition}\label{LemRk}
	Let Assumption \ref{assume_f_compact} holds, then for $1 \leq k \leq N$, we have the following bound for $\mathcal{R}_k$ in \eqref{BBGKY_k1} and \eqref{BBGKY_k}. For arbitrary small $\delta >0$, the following estimate holds for any test function $\Phi\in C^\infty_0(\R^{6k})$,
	\[
	\left|\dotsint (\mathrm{d}q \mathrm{d}p)^{\otimes k}\Phi(q_1, p_1,\dots,q_k,p_k)\nabla_{\vec{q}_k}\cdot \mathcal{R}_k \right| \leq C \hbar^{\frac{1}{2}-\delta},
	\]
	where $C$ depends on $\|D^{s(\delta)}\Phi\|_{\infty}$ and $k$.
\end{Proposition}

\begin{Proposition} \label{lemma_vla_k1_convg}
	Let Assumption \ref{assume_f_compact} and \ref{assume_V} hold, then we have the following bound for $\widetilde{\mathcal{R}}_1$ in \eqref{bbgky_remainder_1}. For arbitrary small $\delta >0$, the following estimate holds for any test function $\Phi\in C^\infty_0(\R^{6})$,
	\begin{equation} \label{lemma_vla_k1_convg_0}
	\begin{aligned}
	\bigg| &  \iint \mathrm{d}q_1 \mathrm{d}p_1 \Phi(q_1, p_1) \nabla_{p_1} \cdot \widetilde{\mathcal{R}}_1 \bigg| \leq C \hbar^{\frac{1}{2}-\delta},
	\end{aligned}
	\end{equation}
	where $C$ depends on $\|D^{s(\delta)}\Phi\|_{\infty}$.
\end{Proposition}

\begin{Proposition} \label{mk_BBFKY_to_infty}
	Suppose that Assumption \ref{assume_f_compact} and \ref{assume_V} hold. Denote the remainders terms $\widetilde{\mathcal{R}}_k$ and $\widehat{\mathcal{R}}_k$ as in \eqref{bbgky_remainder_k}. Then for  $1 \leq k \leq N$ and arbitrary small $\delta >0$, the following estimates hold for any test function $\Phi\in C^\infty_0(\R^{6k})$,
	\begin{equation}\label{mk_BBFKY_to_infty_1}
	\left| \dotsint  (\mathrm{d}q \mathrm{d}p)^{\otimes k} \Phi (q_1, p_1,\dots,q_k,p_k) \cdot \widehat{\mathcal{R}}_k \right|\leq C \hbar^{3-\delta} ,
	\end{equation}
	and
	\begin{equation}\label{mk_BBFKY_to_infty_2}
	\left| \dotsint  (\mathrm{d}q \mathrm{d}p)^{\otimes k} \Phi (q_1, p_1,\dots,q_k,p_k) \nabla_{\vec{p}_k} \cdot \widetilde{\mathcal{R}}_k \right|\leq C \hbar^{\frac{1}{2}-\delta},
	\end{equation}
	where $C$ depends on $\|D^{s(\delta)}\Phi\|_{\infty}$ and $k$.
\end{Proposition}

\subsection{Convergence to infinite hierarchy}

In this subsection, we prove that the $k$-particle Husimi measure $m_{N,t}^{(k)}$ has subsequence that converges weakly (as $N\to \infty$) to a limit $m_{t}^{(k)}$ in $L^1$, which is a solution of the infinite hierarchy in the sense of distribution.

The weak compactness of $k$-particle Husimi measure $m_{N,t}^{(k)}$ can be proved by the use of Dunford-Pettis theorem.\footnote{See \cite{diestel_uniinter} for the treatment of uniform integerability.} In particular, we have the following result.
\begin{Proposition}
	\label{propweaklimit}
	Let $\{m_{N,t}^{(k)}\}_{N\in \N}$ be the $k$-particle Husimi measure, then there exists a subsequence $\{m_{N_j,t}^{(k)}\}_{j\in \N}$ that converges weakly in $L^1(\R^{6k})$ to a function $ (2\pi)^{3k}m_{t}^{(k)}$, i.e. for all $ \varphi\in L^\infty (\R^{6k})$, it holds
	\[
	\frac{1}{(2\pi)^{3k}}\dotsint (\mathrm{d}q\mathrm{d}p)^{\otimes k}  m_{N_j,t}^{(k)} \varphi \rightarrow \dotsint (\mathrm{d}q\mathrm{d}p)^{\otimes k}  m_{t}^{(k)} \varphi,
	\]
	when $j \to \infty$ for arbitrary fixed $k\geq 1$.
\end{Proposition}

\begin{proof}
	To apply Dunford-Pettis theorem, we need to check that it is uniformly integrable and bounded. 
	From the previous uniform estimates that we have obtained for $m_{N,t}^{(k)}$ from \eqref{Linftykestimate} and its second finite moment in Proposition \ref{2nd_moment_finite} imply
	\begin{align*}
	\norm{m_{N,t}^{(k)}}_{L^\infty}\leq 1,\quad  \norm{(|\vec{q}_k|+|\vec{p}_k|)m_{N,t}^{(k)}}_{L^1} \leq C(t).
	\end{align*}
	where $\vec{q}_k := (q_1,\dots,q_k)$, $\vec{p}_k := (p_1,\dots,p_k)$ and $C(t)$ is a time-dependent constant, we can check the uniform integrability. More precisely, for any $\varepsilon > 0$, by taking $r = {\varepsilon}^{-1}{(2\pi)^{3k}C(t)}$ we have that
	\begin{equation}
	\frac{1}{(2\pi)^{3k}}\dotsint_{|\vec{q}_k|+|\vec{p}_k| \geq r} (\mathrm{d}q\mathrm{d}p)^{\otimes k}  m_{N,t}^{(k)} \leq \frac{1}{r}  \frac{1}{(2\pi)^{3k}}\dotsint (\mathrm{d}q\mathrm{d}p)^{\otimes k} \left( |\vec{q}_k|+|\vec{p}_k| \right) m_{N,t}^{(k)} \leq \varepsilon.
	\end{equation}
	Furthermore, for arbitrary $\varepsilon >0 $, by taking $\delta = \varepsilon$, we have that for all $E\subset\R^{6k}$ with $\text{Vol}(E) \leq \delta$, it holds
	\[
	\dotsint_E m^{(k)}_{N,t} \leq \norm{m^{(k)}_{N,t}}_\infty \text{Vol}(E) \leq \varepsilon, 
	\]
	which means that there is no concentration for the $k$-particle Husimi measure.

	It is shown in \eqref{MNkL1} that the boundedness of $k$-particle Husimi measure in $L^1$, i.e.
	\[
	\norm{m_{N,t}^{(k)}}_{L^1} \leq (2\pi)^{3k}.
	\]
	Then applying directly Dunford-Pettis Theorem one obtain that $k$-particle Husimi measure is weakly compact in $L^1$. 
\end{proof}

\begin{proof}[\bf Proof of Theorem \ref{thmmain} and Corollary \ref{cor:1-Wasserstein}]
	Cantor's diagonal procedure shows that we can take the same convergent subsequence of $m_{N,t}^{(k)}$ for all $k\geq 1$. Then by the error estimates obtained in Propositions \ref{LemRk}, \ref{lemma_vla_k1_convg}, 
	and \ref{mk_BBFKY_to_infty}, we can obtain that the limit satisfies the infinite hierarchy \eqref{infinitehierarchy} in the sense of distribution, by directly taking the limit in the weak formulation of \eqref{BBGKY_k1} and \eqref{BBGKY_k}.
	
	Observe that the estimates for the remainder terms also show that any convergent subsequence of $m_{N,t}^{(k)}$ converges weakly in $L^1$ to the solution of the infinite hierarchy. Therefore, if furthermore, the infinite hierarchy has a unique solution, then the sequence $m_{N,t}^{(k)}$ itself { converges} weakly to the solution of the infinite hierarchy.
	
	As for Corollary \ref{cor:1-Wasserstein}, one only need to combine the facts that the infinite hierarchy has a unique solution and that the tensor products of the solution of the Vlasov equation \eqref{classical_Vlasov}, $m_t^{\otimes k}$ is a solution of the infinite hierarchy.
	
	{Lastly, by Theorem 7.12 in \cite{Villani2003}, we would obtain the convergence in $1$-Wasserstein metric.}
\end{proof}

\section{Completion of the reformulation and estimates in the proof}\label{section_proof}

\subsection{Proof of the reformulation in section \ref{sec2.1}}\label{proof_of_reformulation}
In this subsection we supply the proofs for the reformulation of Schr\"odinger equation into a hierarchy of $k$ $(1\leq k\leq N)$ particle Husimi measure. The reformulation shares similar structure to the classical BBGKY hierarchy.

\begin{proof}[Proof of Proposition \ref{lemma_vla_k1}]
	First, observe that taking the time derivative on the Husimi measure, we have
	\begin{align*}
	&2\mathrm{i}\hbar \p_t m^{(1)}_{N,t}(q_1,p_1) \\
	= &\bigg(\hbar^2 \iiint \mathrm{d}w\mathrm{d}u\mathrm{d}x\ f_{q_1,p_1}^\hbar (w) \overline{f_{q_1,p_1}^\hbar (u)} \left<\psi_{N,t}, a^*_w a_u \nabla_x a_x^* \nabla_x a_x \psi_{N,t} \right> \\
	& - \hbar^2  \iiint \mathrm{d}w\mathrm{d}u\mathrm{d}x\ \overline{f_{q_1,p_1}^\hbar (w)}f_{q_1,p_1}^\hbar (u)\left< \psi_{N,t}, \nabla_x a_x^* \nabla_x a_x  a^*_u a_w \psi_{N,t} \right> \bigg)\\
	&+ \bigg(  \frac{1}{N}  \iint \mathrm{d}w\mathrm{d}u\iint \mathrm{d}x\mathrm{d}y\  f_{q_1,p_1}^\hbar (w) \overline{f_{q_1,p_1}^\hbar (u)} \left< \psi_{N,t},  V(x-y) a^*_w a_u a^*_x a^*_y a_y a_x \psi_{N,t} \right>\\
	& -  \frac{1}{N}   \iint \mathrm{d}w\mathrm{d}u\iint \mathrm{d}x\mathrm{d}y\  \overline{f_{q_1,p_1}^\hbar (w)} f_{q_1,p_1}^\hbar (u) \left< \psi_{N,t},  V(x-y) a^*_x a^*_y a_y a_x a^*_u a_w \psi_{N,t} \right> \bigg)\\
	= &:  I_1 + \mathit{II}_1.
	\end{align*}
	Now, focus on $I_1$, we have
	\begin{align*}
	I_1 = & \hbar^2 \iiint \mathrm{d}w\mathrm{d}u\mathrm{d}x\ f_{q_1,p_1}^\hbar (w) \overline{f_{q_1,p_1}^\hbar (u)} \left< \psi_{N,t}, a^*_w a_u \nabla_x a_x^* \nabla_x a_x \psi_{N,t} \right> \\
	& - \hbar^2 \iiint \mathrm{d}w\mathrm{d}u\mathrm{d}x\ f_{q_1,p_1}^\hbar (w) \overline{f_{q_1,p_1}^\hbar (u)}\left<\psi_{N,t}, \nabla_x a_x^* \nabla_x a_x  a^*_w a_u \psi_{N,t} \right>,
	\end{align*}
	where the last equality is just change of variable on the complex conjugate term. Then, from CAR, observe we have that
	\begin{align*}
	- a^*_w a_u a^*_x \Delta_x a_x = & a^*_w a^*_x a_u \Delta_x a_x - \delta_{u=x}a^*_w \Delta_x a_x \\
	= & a^*_x a^*_w \Delta_x a_x a_u -  \delta_{u=x}a^*_w \Delta_x a_x \\
	= & \Delta_x a_x^* a_w^* a_x a_u -  \delta_{u=x}a^*_w \Delta_x a_x \\
	= & - \Delta_x a_x^* a_x a^*_w a_u + \delta_{w=x} \Delta_x a_x^* a_u -  \delta_{u=x}a^*_w \Delta_x a_x,
	\end{align*}
	where integration by parts and CAR of the operator have been used several times. Putting this back, we cancel out the the second term and get
	\begin{equation}\label{Fock_Kinectic_Schro_2}
	\begin{aligned}
	I_1 = & \hbar^2 \iiint \mathrm{d}w\mathrm{d}u\mathrm{d}x\ f_{q_1,p_1}^\hbar (w) \overline{f_{q_1,p_1}^\hbar (u)} \left< \psi_{N,t},  \big( \delta_{w=x} \Delta_x a_x^* a_u -  \delta_{u=x}a^*_w \Delta_x a_x \big) \psi_{N,t} \right> \\
	= &  \hbar^2 \iint \mathrm{d}w\mathrm{d}u\ \bigg( \Delta_w f_{q_1,p_1}^\hbar (w) \bigg) \overline{f_{q_1,p_1}^\hbar (u)} \left< \psi_{N,t},a_w^* a_u \psi_{N,t} \right> \\
	& -  \hbar^2 \iint \mathrm{d}w\mathrm{d}u\  f_{q_1,p_1}^\hbar (w)  \bigg( \Delta_u \overline{f_{q_1,p_1}^\hbar (u)} \bigg) \left<\psi_{N,t} ,a_w^* a_u \psi_{N,t} \right>.
	\end{aligned}
	\end{equation}
	Now, observe the following
	\begin{align*}
	\nabla_u \overline{f^\hbar_{q_1,p_1} (u)} = & \nabla_u \left( \hbar^{-\frac{3}{4}} f \left( \frac{u-q_1}{\sqrt{\hbar}}\right) e^{-\frac{\mathrm{i}}{\hbar} p_1 \cdot u }  \right)\\
	= & \hbar^{-\frac{3}{4}} \nabla_u f \left( \frac{u-q_1}{\sqrt{\hbar}}\right) e^{-\frac{\mathrm{i}}{\hbar} p_1 \cdot u } +  \hbar^{-\frac{3}{4}} f \left( \frac{u-q_1}{\sqrt{\hbar}}\right) \nabla_u e^{-\frac{\mathrm{i}}{\hbar} p_1 \cdot u } \\
	= & - \hbar^{-\frac{3}{4}} \nabla_{q_1} f \left( \frac{u-q_1}{\sqrt{\hbar}}\right) e^{-\frac{\mathrm{i}}{\hbar} p_1 \cdot u } -\mathrm{i} \hbar^{-1} p_1 \cdot \hbar^{-\frac{3}{4}} f \left( \frac{u-q_1}{\sqrt{\hbar}}\right) e^{-\frac{\mathrm{i}}{\hbar} p_1 \cdot u }\\
	= & (- \nabla_{q_1} - \mathrm{i} \hbar^{-1} p_1)  \overline{f^\hbar_{q_1,p_1} (u)},
	\end{align*}
	and furthermore,
	\begin{equation}\label{Fock_Kinectic_Schro_lapcoherent1}
	\begin{aligned}
	\Delta_u  \overline{f^\hbar_{q_1,p_1} (u)} = & \nabla_u \cdot \nabla_u  \overline{f^\hbar_{q_1,p_1} (u)} \\
	=&\nabla_u \cdot (- \nabla_{q_1} - \mathrm{i} \hbar^{-1} p_1)  \overline{f^\hbar_{q_1,p_1} (u)} \\
	=& (- \nabla_{q_1} - \mathrm{i} \hbar^{-1} p_1)\cdot (- \nabla_{q_1} - \mathrm{i} \hbar^{-1} p_1)  \overline{f^\hbar_{q_1,p_1} (u)}\\
	= & \bigg( \Delta_{q_1} + 2 \mathrm{i} \hbar^{-1} p_1 \cdot \nabla_{q_1}  - \hbar^{-2} p_1^2\bigg) \overline{f^\hbar_{q_1,p_1} (u)}.
	\end{aligned}
	\end{equation}
	and similarly
	\begin{equation}\label{Fock_Kinectic_Schro_lapcoherent2}
	\begin{aligned}
	\Delta_w f_{q_1,p_1}^\hbar (w) = \bigg( \Delta_{q_1} - 2 \mathrm{i} \hbar^{-1} p_1 \cdot \nabla_{q_1}  - \hbar^{-2} p_1^2\bigg) f_{q_1,p_1}^\hbar (w),
	\end{aligned}
	\end{equation}
	we obtain by putting these back into \eqref{Fock_Kinectic_Schro_2},
	\begin{equation}\label{Fock_Kinectic_Schro_3}
	\begin{aligned}
	I_1 = & \hbar^2\bigg[ \left< \Delta_{q_1} \int \mathrm{d}w\ \overline{f^\hbar_{q_1,p_1} (w)} a_w \psi_{N,t}, \int \mathrm{d}u\ \overline{f^\hbar_{q_1,p_1} (u)} a_u \psi_{N,t} \right>\\
	& - \left< \int \mathrm{d}w\ \overline{f^\hbar_{q_1,p_1} (w)} a_w \psi_{N,t} ,\Delta_{q_1} \int \mathrm{d}u\ \overline{f^\hbar_{q_1,p_1} (u)}a_u \psi_{N,t} \right> \bigg]\\
	& -  2\mathrm{i} \hbar p_1 \cdot \bigg[ \left< \nabla_{q_1} \int \mathrm{d}w\ \overline{f^\hbar_{q_1,p_1} (w)} a_w \psi_{N,t}, \int \mathrm{d}u\ \overline{f^\hbar_{q_1,p_1} (u)}a_u \psi_{N,t} \right>\\
	& + \left< \int \mathrm{d}w\ \overline{f^\hbar_{q_1,p_1} (w)} a_w \psi_{N,t},\nabla_{q_1} \int \mathrm{d}u\ \overline{f^\hbar_{q_1,p_1} (u)}a_u \psi_{N,t} \right> \bigg]\\
	= &  2 \mathrm{i} \hbar^2 \Im \left< \Delta_{q_1} a(f^\hbar_{q_1,p_1}) \psi_{N,t}, a(f^\hbar_{q_1,p_1}) \psi_{N,t} \right> - 2\mathrm{i} \hbar p_1 \cdot \nabla_{q_1} m_{N,t}^{(1)}(q_1,p_1).
	\end{aligned}
	\end{equation}
	Since the Husimi measure is actually a real-valued function, we have that
	\begin{equation}\label{Fock_Kinectic_Schro_4}
	\p_t m^{(1)}_{N,t}(q_1,p_1) + p_1 \cdot \nabla_{q_1} m^{(1)}_{N,t}(q_1,p_1)= \text{Re} \left( \frac{\mathit{II}_1}{2\mathrm{i}\hbar}\right) +  \hbar \Im \left< \Delta_{q_1} a(f^\hbar_{q_1,p_1}) \psi_{N,t}, a(f^\hbar_{q_1,p_1}) \psi_{N,t} \right>.
	\end{equation}
	
	Now, we turn our focus on $\mathit{II}_1$, i.e.,
	\begin{align*}
	\mathit{II}_1 & = \frac{1}{N} \iint \mathrm{d}w\mathrm{d}u\iint \mathrm{d}x\mathrm{d}y\  f_{q_1,p_1}^\hbar (w) \overline{f_{q_1,p_1}^\hbar (u)} \left< \psi_{N,t},  V(x-y) a^*_w a_u a^*_x a^*_y a_y a_x \psi_{N,t} \right>\\
	& -  \frac{1}{N} \iint \mathrm{d}w\mathrm{d}u\iint \mathrm{d}x\mathrm{d}y\  \overline{f_{q_1,p_1}^\hbar (w)} f_{q_1,p_1}^\hbar (u) \left< \psi_{N,t},  V(x-y) a^*_x a^*_y a_y a_x a^*_u a_w \psi_{N,t} \right>.
	\end{align*}
	Observe that
	\begin{align*}
	a^*_w a_u a^*_x a^*_y a_y a_x = & a^*_x a^*_y a_y a_x a^*_w a_u\\
	& + \delta_{w=y} a^*_x a^*_y a_x a_u - \delta_{w= x} a^*_x a^*_y a_y a_u\\
	&+ \delta_{u=x}a^*_w a^*_y a_y a_x - \delta_{u = y} a^*_w a^*_x a_y a_x.
	\end{align*}
	The first term and the complex conjugate term vanishes under changes of variable, $u$ to $w$ and $w$ to $u$. Therefore, since from assumption $V(x) = V(-x)$, we have
	\begin{equation}\label{Fock_Potential_V_1}
	\begin{aligned}
	\mathit{II}_1 =&  \frac{1}{N} \iiint \mathrm{d}w\mathrm{d}u\mathrm{d}x\  f_{q_1,p_1}^\hbar (w) \overline{f_{q_1,p_1}^\hbar (u)}\left< \psi_{N,t},  V(x-w)  a^*_x a^*_w a_x a_u \psi_{N,t} \right> \\
	& -  \frac{1}{N} \iiint \mathrm{d}w\mathrm{d}u\mathrm{d}x\  f_{q_1,p_1}^\hbar (w) \overline{f_{q_1,p_1}^\hbar (u)} \left< \psi_{N,t},  V(x-u)  a^*_w a^*_x a_u a_x \psi_{N,t} \right>\\
	& +  \frac{1}{N}\iiint \mathrm{d}w\mathrm{d}u\mathrm{d}y\ f_{q_1,p_1}^\hbar (w) \overline{f_{q_1,p_1}^\hbar (u)}\left< \psi_{N,t},  V(u-y)  a^*_w a^*_y a_y a_u \psi_{N,t} \right>\\
	& -   \frac{1}{N}\iiint \mathrm{d}w\mathrm{d}u\mathrm{d}y\ f_{q_1,p_1}^\hbar (w) \overline{f_{q_1,p_1}^\hbar (u)} \left< \psi_{N,t},  V(w-y)   a^*_w a^*_y a_y a_u \psi_{N,t} \right>\\
	= &  \frac{1}{N} \iiint \mathrm{d}w\mathrm{d}u\mathrm{d}x\  f_{q_1,p_1}^\hbar (w) \overline{f_{q_1,p_1}^\hbar (u)} \bigg( V(u-x)- V(w-x) \bigg)\left< \psi_{N,t}, a^*_w a^*_x  a_x a_u \psi_{N,t} \right> \\
	& +  \frac{1}{N}\iiint \mathrm{d}w\mathrm{d}u\mathrm{d}y\ f_{q_1,p_1}^\hbar (w) \overline{f_{q_1,p_1}^\hbar (u)}\bigg( V(u-y)- V(w-y) \bigg)\left< \psi_{N,t},    a^*_w a^*_y a_y a_u\psi_{N,t} \right> \\
	= & \frac{2}{N} \iiint \mathrm{d}w\mathrm{d}u\mathrm{d}y\ f_{q_1,p_1}^\hbar (w) \overline{f_{q_1,p_1}^\hbar (u)}\bigg( V(u-y)- V(w-y) \bigg)\left< \psi_{N,t},    a^*_w a^*_y a_y a_u\psi_{N,t} \right>. 
	\end{aligned}
	\end{equation}
	Now, note that mean value theorem gives
	\begin{equation}\label{MeanValueTheorem_V}
	V(u-y)- V(w-y) = \int_0^1 \mathrm{d}s \nabla V\big(s(u-y)+(1-s)(w-y) \big)\cdot (u-w),
	\end{equation}
	and observe that since, $V\big(s(u-y)+(1-s)(w-y) \big) = V\big(su+(1-s)w - y \big)$, we can have from  \eqref{Fock_Potential_V_1} the following
	\begin{equation}\label{Fock_Potential_V_2}
	\begin{aligned}
	\mathit{II}_1 = & \frac{2}{N} \iiint \mathrm{d}w\mathrm{d}u\mathrm{d}y\ f_{q_1,p_1}^\hbar (w) \overline{f_{q_1,p_1}^\hbar (u)}\left( \int_0^1 \mathrm{d}s \nabla V\big(su+(1-s)w - y \big) \right)\cdot (u-w)\cdot\\
	&\hspace{1cm}\left< \psi_{N,t}, a^*_w a^*_y a_y a_u\psi_{N,t} \right> \\
	= & \frac{\mathrm{2 \mathrm{i}}\hbar}{N} \iiint \mathrm{d}w\mathrm{d}u\mathrm{d}y \int_0^1 \dd{s} \nabla V\big(su+(1-s)w - y \big) \cdot \nabla_{p_1} \left( f_{q_1,p_1}^\hbar (w) \overline{f_{q_1,p_1}^\hbar (u)} \right)\left< \psi_{N,t}, a^*_w a^*_y a_y a_u\psi_{N,t} \right> \\
	= & \frac{\mathrm{2 \mathrm{i}}\hbar}{N} \iiint \mathrm{d}w\mathrm{d}u\mathrm{d}y \int_0^1 \dd{s} \nabla V\big(su+(1-s)w - y \big) \cdot \nabla_{p_1} \left(f_{q_1,p_1}^\hbar (w) \overline{f_{q_1,p_1}^\hbar (u)}\right)\left<a_w a_y \psi_{N,t},  a_u a_y\psi_{N,t} \right>,
	\end{aligned}
	\end{equation}
	where we use the fact that
	\begin{equation}\label{grad_p_fhbar}
	\nabla_{p_1} \left(f_{q_1,p_1}^\hbar (w) \overline{f_{q_1,p_1}^\hbar (u)}\right) = \frac{ \mathrm{i}}{\hbar} (w-u)\cdot f_{q_1,p_1}^\hbar (w) \overline{f_{q_1,p_1}^\hbar (u)}.
	\end{equation}
	Then we get
	\begin{equation}\label{Fock_Potential_V_3}
	\mathit{II}_1  =  \frac{2 \mathrm{i}\hbar}{N} \iiint \mathrm{d}w\mathrm{d}u\mathrm{d}y \int_0^1 \dd{s} \nabla V\big(su+(1-s)w - y \big) \cdot \nabla_{p_1} \left( f_{q_1,p_1}^\hbar (w) \overline{f_{q_1,p_1}^\hbar (u)} \right)\left<a_w a_y \psi_{N,t},  a_u a_y\psi_{N,t} \right>.
	\end{equation}
	Applying the following projection
	\begin{equation} \label{coherent_projection}
	\frac{1}{(2\pi \hbar)^3} \iint \mathrm{d}q_2 \mathrm{d}p_2 \ket{f^\hbar_{q_2,p_2}}\bra{f^\hbar_{q_2,p_2}} =  \mathds{1},
	\end{equation}
	onto $a_y \psi_{N,t}$, we get
	\begin{align*}
	a_y \psi_{N,t} = \frac{1}{(2\pi \hbar)^3} \iint \mathrm{d}q_2 \mathrm{d}p_2\ f^\hbar_{q_2,p_2}(y) \int \mathrm{d}v\ \overline{f^\hbar_{q_2,p_2}(v)} a_v \psi_{N,t} .
	\end{align*}
	Putting this back into \eqref{Fock_Potential_V_3}, we get the following
	\begin{equation}
	\begin{aligned}
	\mathit{II}_1 = & \frac{2 \mathrm{i}\hbar}{N} \frac{1}{(2\pi \hbar)^3} \iint \mathrm{d}w\mathrm{d}u \iint \mathrm{d}y\mathrm{d}v \iint \mathrm{d}q_2 \mathrm{d}p_2 \int_0^1 \mathrm{d}s\ \nabla V\big(su+(1-s)w - y \big) \\ 
	& \hspace{1cm}\cdot \nabla_{p_1} \left( f_{q_1,p_1}^\hbar (w) \overline{f_{q_1,p_1}^\hbar (u)} \right) f_{q_2,p_2}^\hbar (y) \overline{f_{q_2,p_2}^\hbar (v)} \left< a_w a_y \psi_{N,t},  a_u a_v \psi_{N,t} \right>.
	\end{aligned}
	\end{equation}
	Recall that $\hbar^3 = N^{-1}$, we have
	\begin{equation}\label{Fock_Potential_V_4}
	\begin{aligned}
	\mathit{II}_1 = & \frac{2 \mathrm{i}\hbar}{(2\pi)^3} \iint \mathrm{d}w\mathrm{d}u \iint \mathrm{d}y\mathrm{d}v \iint \mathrm{d}q_2 \mathrm{d}p_2 \int_0^1 \mathrm{d}s\ \nabla V\big(su+(1-s)w - y \big) \\ 
	& \hspace{1cm}\cdot \nabla_{p_1} \left( f_{q_1,p_1}^\hbar (w) \overline{f_{q_1,p_1}^\hbar (u)} \right) f_{q_2,p_2}^\hbar (y) \overline{f_{q_2,p_2}^\hbar (v)} \left< a_w a_y \psi_{N,t},  a_u a_v \psi_{N,t} \right>.
	\end{aligned}
	\end{equation}
	Therefore, we have the last term in \eqref{Fock_Kinectic_Schro_4} as
	\begin{align*}
	\Re \frac{\mathit{II}_1}{2i\hbar} =&   \frac{1}{(2\pi)^3} \Re \iint \mathrm{d}w\mathrm{d}u \iint \mathrm{d}y\mathrm{d}v \iint \mathrm{d}q_2 \mathrm{d}p_2 \int_0^1 \mathrm{d}s\ \nabla V\big(su+(1-s)w - y \big) \\ 
	& \hspace{1cm}\cdot \nabla_{p_1} \left( f_{q_1,p_1}^\hbar (w) \overline{f_{q_1,p_1}^\hbar (u)} \right) f_{q_2,p_2}^\hbar (y) \overline{f_{q_2,p_2}^\hbar (v)} \left< a_w a_y \psi_{N,t},  a_u a_v \psi_{N,t} \right>,
	\end{align*}
	thus we have derived the equation for $m_{N,t}^{(1)}(q_1,p_1)$.
\end{proof}

{We have proved the reformulation from Schr\"odinger equation into $1$-particle Husimi measure.
	We also observed that it contains a resemblance to the classical Vlasov equation. Next we want to prove the similar result for $2 \leq k \leq N$.}

\begin{proof}[Proof of Proposition \ref{lemma_vla_bbgky_hierarchy}]
	Now we focus on the case where $2\leq k \leq N$. As in the proof for the case of $k=1$, we first observe that for every $k \in \N$,
	\begin{align} 
	&2 \mathrm{i} \hbar \p_t m_{N,t}^{(k)}(q_1,p_1,\dots,q_k,p_k) \notag\\
	=& \bigg( -\hbar^2 \dotsint (\mathrm{d}w\mathrm{d}u)^{\otimes k} \int \mathrm{d}x \left( f_{q,p}^\hbar (w) \overline{f_{q,p}^\hbar (u)} \right)^{\otimes k} \Delta_x  \left< \psi_{N,t}, a^*_{w_1} \cdots a^*_{w_k} a_{u_k}\cdots a_{u_1} a^*_x a_x \psi_{N,t} \right> \notag\\
	&+ \hbar^2 \dotsint (\mathrm{d}w\mathrm{d}u)^{\otimes k} \int \mathrm{d}x \left( f_{q,p}^\hbar (w) \overline{f_{q,p}^\hbar (u)} \right)^{\otimes k} \Delta_x  \left< \psi_{N,t}, a^*_x a_x a^*_{w_1} \cdots a^*_{w_k} a_{u_k}\cdots a_{u_1} \psi_{N,t} \right> \bigg) \notag\\
	&+ \bigg( \frac{1}{N} \dotsint (\mathrm{d}w\mathrm{d}u)^{\otimes k} \iint \mathrm{d}x\mathrm{d}y V(x-y) \left( f_{q,p}^\hbar (w) \overline{f_{q,p}^\hbar (u)} \right)^{\otimes k}  \left< \psi_{N,t}, a^*_{w_1} \cdots a^*_{w_k} a_{u_k}\cdots a_{u_1} a^*_x a^*_y a_y a_x \psi_{N,t} \right> \notag\\
	&-  \frac{1}{N} \dotsint (\mathrm{d}w\mathrm{d}u)^{\otimes k} \iint \mathrm{d}x\mathrm{d}y V(x-y) \left( f_{q,p}^\hbar (w) \overline{f_{q,p}^\hbar (u)} \right)^{\otimes k}  \left< \psi_{N,t}, a^*_x a^*_y a_y a_x a^*_{w_1} \cdots a^*_{w_k} a_{u_k}\cdots a_{u_1} \psi_{N,t} \right> \bigg) \notag\\
	=: &\  I_2 + \mathit{II}_2,\label{mk_proof_main}
	\end{align}
	where the tensor product denotes $(\mathrm{d}w\mathrm{d}u)^{\otimes k} = \mathrm{d}w_1 \cdots \mathrm{d}w_k \mathrm{d}u_1 \cdots \mathrm{d}u_k$.
	
	We first focus on the $I_2$ part of \eqref{mk_proof_main}, i.e.,
	\begin{equation} \label{mk_kinetic_1}
	\begin{aligned}
	I_2 = & -\hbar^2 \dotsint (\mathrm{d}w\mathrm{d}u)^{\otimes k} \int \mathrm{d}x \left( f_{q,p}^\hbar (w) \overline{f_{q,p}^\hbar (u)} \right)^{\otimes k} \Delta_x  \left< \psi_{N,t}, a^*_{w_1} \cdots a^*_{w_k} a_{u_k}\cdots a_{u_1} a^*_x a_x \psi_{N,t} \right> \\
	& + \hbar^2 \dotsint (\mathrm{d}w\mathrm{d}u)^{\otimes k} \int \mathrm{d}x \left( f_{q,p}^\hbar (w) \overline{f_{q,p}^\hbar (u)} \right)^{\otimes k} \Delta_x \left< \psi_{N,t}, a^*_x a_x a^*_{w_1} \cdots a^*_{w_k} a_{u_k}\cdots a_{u_1} \psi_{N,t} \right>.
	\end{aligned}
	\end{equation}
	Observe that we have
	\begin{equation} \label{anil_magic_K}
	\begin{aligned}
	a^*_{w_1} \cdots a^*_{w_k} a_{u_k}\cdots a_{u_1} a^*_x a_x =&  (-1)^{4k} a^*_x a_x a^*_{w_1} \cdots a^*_{w_k} a_{u_k}\cdots a_{u_1}\\ 
	&+ a^*_x \left( \sum_{j=1}^k (-1)^j \delta_{x=w_j} a^*_{w_1} \cdots \widehat{a^*_{w_j}} \cdots a^*_{w_k} \right) a_{u_k}\cdots a_{u_1} \\
	&- a^*_{w_1} \cdots a^*_{w_k}\left( \sum_{j=1}^k (-1)^j \delta_{x=u_j}  a_{u_k}\cdots \widehat{a_{u_j}}\cdots a_{u_1}  \right) a_x,
	\end{aligned}
	\end{equation}
	where the \textit{hat} indicates exclusion of that element. 
	
	Putting this back into \eqref{mk_kinetic_1}, we obtain
	\begin{equation} \label{mk_kinetic_2_aa}
	\begin{aligned}
	I_2 = & \hbar^2 \dotsint (\mathrm{d}w\mathrm{d}u)^{\otimes k} \int \mathrm{d}x \left( f_{q,p}^\hbar (w) \overline{f_{q,p}^\hbar (u)} \right)^{\otimes k}\\
	&\hspace{2cm} \cdot\Delta_x   \left<  \psi_{N,t}, a^*_{w_1} \cdots a^*_{w_k}\left( \sum_{j=1}^k (-1)^j \delta_{x=u_j}  a_{u_k}\cdots \widehat{a_{u_j}}\cdots a_{u_1}  \right) a_x  \psi_{N,t} \right>\\
	& - \hbar^2 \dotsint (\mathrm{d}w\mathrm{d}u)^{\otimes k} \int \mathrm{d}x \left( f_{q,p}^\hbar (w) \overline{f_{q,p}^\hbar (u)} \right)^{\otimes k}\\
	&\hspace{2cm} \cdot \Delta_x  \left<  \psi_{N,t},a^*_x \left( \sum_{j=1}^k (-1)^j \delta_{x=w_j} a^*_{w_1} \cdots \widehat{a^*_{w_j}} \cdots a^*_{w_k} \right) a_{u_k}\cdots a_{u_1}  \psi_{N,t} \right>\\
	= &   \hbar^2 \sum_{j=1}^k (-1)^j \dotsint (\mathrm{d}w\mathrm{d}u)^{\otimes k} \left( f_{q,p}^\hbar (w) \overline{f_{q,p}^\hbar (u)} \right)^{\otimes k}\\
	&\hspace{2cm}\cdot \bigg( \Delta_{u_j}  \left<  \psi_{N,t}, a^*_{w_1} \cdots a^*_{w_k} \left( a_{u_k}\cdots \widehat{a_{u_j}}\cdots a_{u_1} \right) a_{u_j}  \psi_{N,t} \right>\\
	&\hspace{3cm} - \Delta_{w_j} \left<  \psi_{N,t}, a^*_{w_j} \left(a^*_{w_1} \cdots \widehat{a^*_{w_j}} \cdots a^*_{w_k}\right)  a_{u_k}\cdots a_{u_1}  \psi_{N,t} \right> \bigg).
	\end{aligned}
	\end{equation}
	
	Note that, if we want to move the missing $a_{u_j}$ or $a^*_{w_j}$ back to their original position after applying the delta function, we have for fixed $j$
	\begin{align*}
	(-1)^j  a^*_{w_1} \cdots a^*_{w_k} \left[ a_{u_k}\cdots \widehat{a_{u_j}}\cdots a_{u_1} \right] a_{u_j} = & \frac{(-1)^j }{(-1)^{j-1}}  a^*_{w_1} \cdots a^*_{w_k}  a_{u_k}\cdots a_{u_1} \\
	= & (-1)^1 a^*_{w_1} \cdots a^*_{w_k}  a_{u_k}\cdots a_{u_1}, \\
	(-1)^j a^*_{w_j} \left[a^*_{w_1} \cdots \widehat{a^*_{w_j}} \cdots a^*_{w_k}\right]  a_{u_k}\cdots a_{u_1} =& (-1)^1 a^*_{w_1} \cdots  a^*_{w_k}  a_{u_k}\cdots a_{u_1}.
	\end{align*}
	Therefore, continuing from \eqref{mk_kinetic_2_aa}, we have 
	\begin{equation} \label{mk_kinetic_2}
	\begin{aligned}
	I_2 =  - \hbar^2 \sum_{j=1}^k \dotsint (\mathrm{d}w\mathrm{d}u)^{\otimes k} \left( f_{q,p}^\hbar (w) \overline{f_{q,p}^\hbar (u)} \right)^{\otimes k} \left[ \Delta_{u_j} - \Delta_{w_j} \right]   \left< \psi_{N,t}, a^*_{w_1} \cdots  a^*_{w_k}  a_{u_k}\cdots a_{u_1} \psi_{N,t} \right>.
	\end{aligned}
	\end{equation}
	
	Now, by integration by parts on \eqref{mk_kinetic_2} and note that the Laplacian acting on the coherent state would be similar to  \eqref{Fock_Kinectic_Schro_lapcoherent1} and \eqref{Fock_Kinectic_Schro_lapcoherent2}, i.e., for fixed $j$ where $1\leq j \leq k$
	\begin{align*}
	\Delta_{u_j}  \left(\overline{f_{q,p}^\hbar (u)} \right)^{\otimes k} =& \left( \Delta_{q_j} + 2\mathrm{i}\hbar^{-1} p_j \cdot \nabla_{q_j} - \hbar^{-2} p_j^2 \right)  \left(  \overline{f_{q,p}^\hbar (u)} \right)^{\otimes k},\\
	\Delta_{w_j}  \left( f_{q,p}^\hbar (w)  \right)^{\otimes k} =& \left( \Delta_{q_j} - 2\mathrm{i}\hbar^{-1} p_j \cdot \nabla_{q_j} - \hbar^{-2} p_j^2 \right)  \left( f_{q,p}^\hbar (w) \right)^{\otimes k}.
	\end{align*}
	Thus, we have similar for when $k=1$, the kinetic part as
	\begin{equation} \label{mk_kinetic_3}
	\begin{aligned}
	I_2 =& -  2\mathrm{i}\hbar \sum_{j=1}^k p_j \cdot \nabla_{q_j} \dotsint (\mathrm{d}w\mathrm{d}u)^{\otimes k}  \left( f_{q,p}^\hbar (w) \overline{f_{q,p}^\hbar (u)} \right)^{\otimes k}  \left< \psi_{N,t}, a^*_{w_1} \cdots a^*_{w_k}a_{u_k} \cdots a_{u_1}\psi_{N,t} \right>\\
	& + 2 \hbar^2 \Im  \sum_{j=1}^k \left< \Delta_{q_j} a\left(f_{q_k,p_k}^\hbar\right)\cdots a\left(f_{q_1,p_1}^\hbar\right) \psi_{N,t},  a\left(f_{q_k,p_k}^\hbar\right)\cdots a\left(f_{q_1,p_1}^\hbar\right) \psi_{N,t} \right>\\
	= &- 2 \mathrm{i}\hbar \vec{p}_k \cdot \nabla_{\vec{q}_k} \left< a\left(f_{q_k,p_k}^\hbar\right)\cdots a\left(f_{q_1,p_1}^\hbar\right) \psi_{N,t},  a\left(f_{q_k,p_k}^\hbar\right)\cdots a\left(f_{q_1,p_1}^\hbar\right) \psi_{N,t} \right>\\
	& + 2 \mathrm{i} \hbar^2 \Im  \sum_{j=1}^k \left< \Delta_{q_j} a\left(f_{q_k,p_k}^\hbar\right)\cdots a\left(f_{q_1,p_1}^\hbar\right) \psi_{N,t},  a\left(f_{q_k,p_k}^\hbar\right)\cdots a\left(f_{q_1,p_1}^\hbar\right) \psi_{N,t} \right>.
	\end{aligned}
	\end{equation}
	Therefore it follows that
	\begin{equation}\label{mk_kinetic_Vlasov_1}
	\begin{aligned}
	I_2  &= - 2 \mathrm{i}\hbar \vec{p}_k \cdot \nabla_{\vec{q}_k} m^{(k)}_{N,t}(q_1,p_1,\dots,q_k,p_k)\\
	&\qquad + 2 \mathrm{i} \hbar^2 \Im  \sum_{j=1}^k \left< \Delta_{q_j} a\left(f_{q_k,p_k}^\hbar\right)\cdots a\left(f_{q_1,p_1}^\hbar\right) \psi_{N,t},  a\left(f_{q_k,p_k}^\hbar\right)\cdots a\left(f_{q_1,p_1}^\hbar\right) \psi_{N,t} \right>.
	\end{aligned}
	\end{equation}
	Now, we turn our focus on part $\mathit{II}_2$ of \eqref{mk_proof_main},
	\begin{equation} \label{mk_potential_1}
	\begin{aligned}
	&\mathit{II}_2 \\
	= & \frac{1}{N} \dotsint (\mathrm{d}w\mathrm{d}u)^{\otimes k}  \left( f_{q,p}^\hbar (w) \overline{f_{q,p}^\hbar (u)} \right)^{\otimes k} \iint \mathrm{d}x\mathrm{d}y\ V(x-y)  \left< \psi, a^*_{w_1} \cdots a^*_{w_k} a_{u_k}\cdots a_{u_1} a^*_x a^*_y a_y a_x \psi \right> \\
	& -  \frac{1}{N} \dotsint (\mathrm{d}w\mathrm{d}u)^{\otimes k} \left( f_{q,p}^\hbar (w) \overline{f_{q,p}^\hbar (u)} \right)^{\otimes k} \iint \mathrm{d}x\mathrm{d}y\ V(x-y)  \left< \psi, a^*_x a^*_y a_y a_x a^*_{w_1} \cdots a^*_{w_k} a_{u_k}\cdots a_{u_1} \psi \right>.
	\end{aligned}
	\end{equation}
	For $1\leq k \leq N$, observe that from the CAR, we have
	\begin{equation} \label{mk_potential_CAR}
	\begin{aligned}
	& a^*_{w_1} \cdots  a^*_{w_k} a_{u_k}\cdots a_{u_1} a^*_x a^*_y a_y a_x - (-1)^{8k} a^*_x a^*_y a_y a_x  a^*_{w_1} \cdots a^*_{w_k} a_{u_k}\cdots a_{u_1} \\
	&= - a^*_{w_1} \cdots a^*_{w_k} \left( \sum_{j=1}^k (-1)^j \delta_{x=u_j} a_{u_k}\cdots \widehat{a_{u_j}}\cdots a_{u_1} \right)a^*_y a_y a_x  \\
	&\quad  - a^*_x a^*_{w_1} \cdots a^*_{w_k}\left( \sum_{j=1}^k (-1)^j \delta_{y=u_j} a_{u_k}\cdots \widehat{a_{u_j}}\cdots a_{u_1} \right) a_y a_x \\
	&\quad + a^*_x a^*_y \left( \sum_{j=1}^k (-1)^j  \delta_{y=w_j}  a^*_{w_1} \cdots \widehat{a^*_{w_j}} \cdots a^*_{w_k} \right) a_{u_k}\cdots a_{u_1} a_x \\
	&\quad + a^*_x a^*_y a_y \left(\sum_{j=1}^k (-1)^j  \delta_{x=w_j}  a^*_{w_1} \cdots \widehat{a^*_{w_j}} \cdots a^*_{w_k} \right) a_{u_k}\cdots a_{u_1}.
	\end{aligned}
	\end{equation}
	From \eqref{mk_potential_1}, we have that
	\begin{align*}
	& \iint \mathrm{d}x \mathrm{d}y\ V(x-y) \left(a^*_{w_1} \cdots  a^*_{w_k} a_{u_k}\cdots a_{u_1} a^*_x a^*_y a_y a_x - a^*_x a^*_y a_y a_x  a^*_{w_1} \cdots a^*_{w_k} a_{u_k}\cdots a_{u_1} \right) \\
	&
	=\iint \mathrm{d}x\mathrm{d}y V(x-y) \bigg[- a^*_{w_1} \cdots a^*_{w_k} \left( \sum_{j=1}^k (-1)^j \delta_{x=u_j} a_{u_k}\cdots \widehat{a_{u_j}}\cdots a_{u_1} \right) a^*_y a_y a_x \\
	&\quad - a^*_x a^*_{w_1} \cdots a^*_{w_k} \left( \sum_{j=1}^k (-1)^j \delta_{y=u_j} a_{u_k}\cdots \widehat{a_{u_j}}\cdots a_{u_1} \right) a_y a_x \\
	&\quad + a^*_x a^*_y a_y \left( \sum_{j=1}^k (-1)^j  \delta_{x=w_j}  a^*_{w_1} \cdots \widehat{a^*_{w_j}} \cdots a^*_{w_k} \right) a_{u_k}\cdots a_{u_1}\\
	&\quad + a^*_x a^*_y a_y \left(\sum_{j=1}^k (-1)^j  \delta_{x=w_j}  a^*_{w_1} \cdots \widehat{a^*_{w_j}} \cdots a^*_{w_k} \right) a_{u_k}\cdots a_{u_1} \bigg]\\
	&
	=: J_1 + J_2 + J_3 + J_4.
	\end{align*}
	Note that summing $J_1$ and $J_4$, we have
	\begin{align*}
	J_1 + J_4 & = - \sum_{j=1}^k (-1)^j  \int \mathrm{d}y \bigg[ \left(V(u_j -y)  a^*_{w_1} \cdots a^*_{w_k} a_{u_k}\cdots \widehat{a_{u_j}}\cdots a_{u_1} a^*_y a_y a_{u_j}  \right)\\
	&\quad - \left(V(w_j- y) a^*_{w_j} a^*_y a_y a^*_{w_1} \cdots \widehat{a^*_{w_j}} \cdots a^*_{w_k} a_{u_k}\cdots a_{u_1} \right) \bigg]\\
	&
	=  \sum_{j=1}^k \left[\int \mathrm{d}y V(u_j -y)  a^*_{w_1} \cdots a^*_{w_k} a_{u_k}\cdots a_{u_1} a^*_y a_y - V(0) a^*_{w_1} \cdots a^*_{w_k} a_{u_k}\cdots a_{u_1} \right]\\
	&\quad -  \sum_{j=1}^k \left[\int \mathrm{d}y V(w_j -y)  a^*_y a_y a^*_{w_1} \cdots a^*_{w_k} a_{u_k}\cdots a_{u_1} - V(0) a^*_{w_1} \cdots a^*_{w_k} a_{u_k}\cdots a_{u_1} \right],
	\intertext{where the terms with $V(0)$ cancel one another. For the remaining term, we use again CAR to obtain}
	&
	= \sum_{j=1}^k \int \mathrm{d}y \big(V(u_j -y)- V(w_j -y) \big)a^*_y a^*_{w_1} \cdots a^*_{w_k} a_{u_k}\cdots a_{u_1}a_y \\
	&\quad+ \sum_{j=1}^k \sum_{i=1}^k (-1)^i \int \mathrm{d}y\ V(u_j -y) \delta_{u_i = y} a^*_{w_1} \cdots a^*_{w_k} a_{u_k}\cdots \widehat{a_{u_i}}\cdots a_{u_1} a_y\\
	&\quad - \sum_{j=1}^k \sum_{i=1}^k (-1)^i \int \mathrm{d}y\ V(w_j -y) \delta_{w_i = y} a^*_y a^*_{w_1} \cdots \widehat{a^*_{w_i}} \cdots a^*_{w_k} a_{u_k}\cdots a_{u_1}\\
	&
	= \sum_{j=1}^k \int \mathrm{d}y \big(V(u_j -y)- V(w_j -y) \big)a^*_y a^*_{w_1} \cdots a^*_{w_k} a_{u_k}\cdots a_{u_1}a_y \\
	&\quad- \sum_{j=1}^k \sum_{i=1}^k \big( V(u_j - u_i) - V(w_j - w_i)\big) a^*_{w_1} \cdots a^*_{w_k} a_{u_k}\cdots a_{u_1}.
	\end{align*}
	On the other hand, the sum of $J_2$ and $J_2$ yield
	\begin{align*}
	J_2 + J_3 &= \sum_{j=1}^k \int \mathrm{d}x \big(V(x-u_j) - V(x-w_j) \big) a^*_x a^*_{w_1} \cdots a^*_{w_k} a_{u_k}\cdots a_{u_1} a_x.
	\end{align*}
	By change of variable and using the fact that $V(-x) = V(x)$, we have from \eqref{mk_potential_1} that
	\begin{equation} \label{mk_potential_2}
	\begin{aligned}
	\mathit{II}_2 =& \frac{2}{N}  \dotsint (\mathrm{d}w\mathrm{d}u)^{\otimes k} \int \mathrm{d}y\  \sum_{j=1}^k \bigg[ V(y-u_j) - V(w_j -y) \bigg] \left( f_{q,p}^\hbar (w) \overline{f_{q,p}^\hbar (u)} \right)^{\otimes k}\\
	&\quad\cdot \left<a_{w_k} \cdots a_{w_1} a_y \psi_{N,t}, a_{u_k} \cdots a_{u_1} a_y \psi_{N,t} \right> \\
	& - \frac{1}{N}  \dotsint (\mathrm{d}w\mathrm{d}u)^{\otimes k} \sum_{j\neq i}^k  \bigg[ V(u_j-u_i) - V(w_j -w_i) \bigg] \left( f_{q,p}^\hbar (w) \overline{f_{q,p}^\hbar (u)} \right)^{\otimes k}\\
	&\quad\cdot \left< a_{w_k} \cdots a_{w_1} \psi_{N,t}, a_{u_k} \cdots a_{u_1} \psi_{N,t} \right> \\
	\end{aligned}
	\end{equation}
	Applying mean value theorem on the first term on right hand side, we have that
	\begin{equation} \label{mk_potential_3}
	\begin{aligned}
	\frac{2}{N} &\sum_{j=1}^k  \dotsint (\mathrm{d}w\mathrm{d}u)^{\otimes k} \int \mathrm{d}y \left( V(y-u_j) - V(w_j -y) \right) \left( f_{q,p}^\hbar (w) \overline{f_{q,p}^\hbar (u)} \right)^{\otimes k}\\
	&\quad\cdot\left< a_{w_k} \cdots a_{w_1} a_y \psi_{N,t}, a_{u_k} \cdots a_{u_1} a_y \psi_{N,t} \right> \\
	= &\frac{2}{N} \sum_{j=1}^k  \dotsint (\mathrm{d}w\mathrm{d}u)^{\otimes k} \int \mathrm{d}y \left[ \int_0^1 \mathrm{d}s\ \nabla V(su_j + (1-s)w_j -y) \right]\cdot (u_j - w_j) \left( f_{q,p}^\hbar (w) \overline{f_{q,p}^\hbar (u)} \right)^{\otimes k}\\
	&\quad\cdot \left< a_{w_k} \cdots a_{w_1} a_y \psi_{N,t}, a_{u_k} \cdots a_{u_1} a_y \psi_{N,t} \right> \\
	= & \frac{2 \mathrm{i} \hbar}{N} \sum_{j=1}^k  \dotsint (\mathrm{d}w\mathrm{d}u)^{\otimes k} \int \mathrm{d}y \left[ \int_0^1 \mathrm{d}s\ \nabla V(su_j + (1-s)w_j -y) \right]\cdot\nabla_{p_j} \left( f_{q,p}^\hbar (w) \overline{f_{q,p}^\hbar (u)} \right)^{\otimes k}\\
	&\quad\cdot \left< a_{w_k} \cdots a_{w_1} a_y \psi_{N,t}, a_{u_k} \cdots a_{u_1} a_y \psi_{N,t} \right>.
	\end{aligned}
	\end{equation}
	As in the case of $k=1$, we apply the projection \eqref{coherent_projection} onto $a_y \psi_{N,t}$ and get further
	\begin{equation} \label{mk_potential_projection}
	\begin{aligned}
	&\frac{2\mathrm{i}\hbar}{N} \sum_{j=1}^k\dotsint (\mathrm{d}w\mathrm{d}u)^{\otimes k} \int \mathrm{d}y \left[ \int_0^1 \mathrm{d}s \nabla  V(su_j + (1-s)w_j -y) \right]\cdot\nabla_{p_j} \left( f_{q,p}^\hbar (w) \overline{f_{q,p}^\hbar (u)} \right)^{\otimes k}\\
	&\quad\cdot \left< a_{w_k} \cdots a_{w_1} a_y \psi_{N,t} , a_{u_k} \cdots a_{u_1} \mathds{1} a_y \psi_{N,t}  \right>\\
	= & \frac{2\mathrm{i}\hbar}{N} \frac{1}{(2\pi \hbar)^3} \sum_{j=1}^k  \dotsint (\mathrm{d}w\mathrm{d}u)^{\otimes k} \int \mathrm{d}y \left[ \int_0^1 \mathrm{d}s \nabla V(su_j + (1-s)w_j -y) \right]\cdot\nabla_{p_j} \left( f_{q,p}^\hbar (w) \overline{f_{q,p}^\hbar (u)} \right)^{\otimes k} \\
	&\quad\cdot \iint \mathrm{d}\widetilde{q} d \widetilde{p}\ f^\hbar_{\widetilde{q}, \widetilde{p}} (y)\int \mathrm{d}v\ \overline{f^\hbar_{\widetilde{q}, \widetilde{p}}(v)} \left< a_{w_k} \cdots a_{w_1} a_y \psi_{N,t} , a_{u_k} \cdots a_{u_1} a_v \psi_{N,t}  \right>.
	\end{aligned}
	\end{equation}		
	Therefore, dividing both equations by $2\mathrm{i}\hbar$,  we have the following equation
	\begin{align}
	&\p_t m_{N,t}^{(k)} (q_1,p_1,\dots,q_k,p_k) + \vec{p}_k \cdot \nabla_{\vec{q}_k} m_{N,t}^{(k)} (q_1,p_1,\dots,q_k,p_k) \notag\\
	=&\hbar \Im  \sum_{j=1}^k \left< \Delta_{q_j} a\left(f_{q_k,p_k}^\hbar\right)\cdots a\left(f_{q_1,p_1}^\hbar\right) \psi_{N,t},  a\left(f_{q_k,p_k}^\hbar\right)\cdots a\left(f_{q_1,p_1}^\hbar\right) \psi_{N,t} \right>\notag\\
	&\quad+\frac{1}{(2\pi)^3} \sum_{j=1}^k  \dotsint (\mathrm{d}w\mathrm{d}u)^{\otimes k} \int \mathrm{d}y \left[ \int_0^1 \mathrm{d}s \nabla V(su_j + (1-s)w_j -y) \right]\cdot\nabla_{p_j} \left( f_{q,p}^\hbar (w) \overline{f_{q,p}^\hbar (u)} \right)^{\otimes k} \notag\\
	&\quad\quad\cdot \iint \mathrm{d}\widetilde{q} d \widetilde{p}\ f^\hbar_{\widetilde{q}, \widetilde{p}} (y)\int \mathrm{d}v\ \overline{f^\hbar_{\widetilde{q}, \widetilde{p}}(v)} \left< a_{w_k} \cdots a_{w_1} a_y \psi_{N,t} , a_{u_k} \cdots a_{u_1} a_v \psi_{N,t}  \right>\notag\\
	&\quad + \frac{\mathrm{i}\hbar^2}{2}  \dotsint (\mathrm{d}w\mathrm{d}u)^{\otimes k} \sum_{j\neq i}^k  \bigg[ V(u_j-u_i) - V(w_j -w_i) \bigg] \left( f_{q,p}^\hbar (w) \overline{f_{q,p}^\hbar (u)} \right)^{\otimes k}\notag\\
	&\quad\quad\cdot \left< a_{w_k} \cdots a_{w_1} \psi_{N,t}, a_{u_k} \cdots a_{u_1} \psi_{N,t} \right>.\label{BBGKY}
	\end{align}
	for $1\leq k \leq N$, $\vec{p}_k = (p_1,\dots,p_k)$ and recalling $\hbar^3 = N^{-1}$. At this point we finish the computation of the hierarchy for Husimi measure.
\end{proof}

\subsection{Proof of the uniform estimates in section \ref{sec2.3}}\label{proof_of_estimations}
This subsection provide the proof of estimates for the error terms that appeared in the equations for $m^{(k)}_{N,t}$. Note that in all the proofs below, we suppose, without loss of generality, that the test function $\Phi \in C^\infty_0(\R^{6k})$ is factorized in phase-space by family of test functions in $C^\infty_0(\R^{3})$ space.
\subsubsection{Proof of Proposition \ref{LemRk}}
\begin{proof}
	For fixed $k$, we denote the vector $\vec{x}_k = (x_1,\cdots, x_k)$ for each $x_j \in \R^3$ with $j=1,\cdots,k$.
	Then we estimate the integral as follows
	\begin{align}
	&\Bigg| \dotsint  (\mathrm{d}q\mathrm{d}p)^{\otimes k}  \nabla_{\vec{q}_k}\Phi(q_1,p_1,\dots,q_k,p_k) \cdot \mathcal{R}_k \Bigg|\notag\\
	\leq &\hbar \Bigg| \sum_{j=1}^k \dotsint (\mathrm{d}q\mathrm{d}p)^{\otimes k}\  \nabla_{q_j} \Phi(q_1,p_1,\dots,q_k,p_k) \notag\\
	&\qquad\cdot \left< \nabla_{q_j}\big( a (f^\hbar_{q_k,p_k}) \cdots a (f^\hbar_{q_1,p_1})\big) \psi_{N,t},  a (f^\hbar_{q_k,p_k}) \cdots a (f^\hbar_{q_1,p_1}) \psi_{N,t} \right> \Bigg|\notag\\
	= & \hbar^{1-\frac{3}{2}k} \bigg|  \sum_{j=1}^k \dotsint (\mathrm{d}q\mathrm{d}p)^{\otimes k} \  \nabla_{q_j} \Phi(q_1,p_1,\dots,q_k,p_k) \cdot \dotsint (\mathrm{d}w\mathrm{d}u)^{\otimes k} \prod_{n=1}^k\left(\rchi_{(w_n-u_n)\in ({\Omega_\hbar^\alpha})^c} +\rchi_{(w_n-u_n)\in {\Omega_\hbar^\alpha}} \right)\notag\\
	&\quad\cdot  \nabla_{q_j} f\left( \frac{w_n-q_n}{\sqrt{\hbar}}\right) f\left( \frac{u_n-q_n}{\sqrt{\hbar}}\right) e^{\frac{i}{\hbar}p_n \cdot(w_n-u_n)}\left<a_{w_k}\cdots a_{w_1} \psi_{N,t}, a_{u_k}\cdots a_{u_1} \psi_{N,t} \right> \bigg|\notag\\
	\leq &  \hbar^{1-\frac{3}{2}k}  \sum_{j=1}^k  \dotsint (\mathrm{d}q\mathrm{d}w\mathrm{d}u)^{\otimes k} \Bigg|\dotsint (\mathrm{d}p)^{\otimes k}  \prod_{n=1}^k \left(\rchi_{(w_n-u_n)\in {\Omega_\hbar^\alpha}} +\rchi_{(w_n-u_n)\in ({\Omega_\hbar^\alpha})^c} \right)   \nabla_{q_j} \Phi\cdot e^{\frac{i}{\hbar}p_n \cdot (w_n -u_n)}\Bigg|\notag\\
	&\quad\cdot \left|\nabla_{q_j}f\left( \frac{w_n-q_n}{\sqrt{\hbar}}\right) \right|\left| f\left( \frac{u_n-q_n}{\sqrt{\hbar}}\right) \right| \norm{a_{w_k}\cdots a_{w_1} \psi_{N,t}}\norm{ a_{u_k}\cdots a_{u_1} \psi_{N,t}}\notag\\
	=&  \hbar^{\frac{1}{2}-\frac{3}{2}k}  \sum_{j=1}^k  \dotsint (\mathrm{d}q\mathrm{d}w\mathrm{d}u)^{\otimes k} \Bigg|\dotsint (\mathrm{d}p)^{\otimes k}  \prod_{n=1}^k \left(\rchi_{(w_n-u_n)\in {\Omega_\hbar^\alpha}} +\rchi_{(w_n-u_n)\in ({\Omega_\hbar^\alpha})^c} \right)\nabla_{q_j} \Phi\cdot e^{\frac{i}{\hbar}p_n \cdot (w_n -u_n)}\Bigg|\notag\\
	&\quad\cdot\prod_{n\neq j}^k  \left|f\left( \frac{w_n-q_n}{\sqrt{\hbar}}\right)  f\left( \frac{u_n-q_n}{\sqrt{\hbar}}\right) \right|\left|\nabla f\left( \frac{w_j-q_j}{\sqrt{\hbar}}\right) \right|\left| f\left( \frac{u_j-q_j}{\sqrt{\hbar}}\right) \right|\notag\\
	&\qquad \cdot \norm{a_{w_k}\cdots a_{w_1} \psi_{N,t}}\norm{ a_{u_k}\cdots a_{u_1} \psi_{N,t}},\label{LemRk_proof_0}
	\end{align}
	where ${\Omega_\hbar^\alpha}$ is defined as in \eqref{estimate_oscillation_omega} and used the fact that
	\[
	\nabla_{q_j} f\left( \frac{w_j-q_j}{\sqrt{\hbar}}\right) = -\frac{1}{\sqrt{\hbar}} \nabla f\left( \frac{w_j-q_j}{\sqrt{\hbar}}\right).
	\]
	Now, the product term $\prod_{n=1}^k \left(\rchi_{(w_n-u_n)\in {\Omega_\hbar^\alpha}} +\rchi_{(w_n-u_n)\in ({\Omega_\hbar^\alpha})^c} \right)$ in \eqref{LemRk_proof_0} includes a summation of $C(k)$ terms of the following type
	\begin{equation}\label{LemRk_proof_1}
	\rchi_{(w_1-u_1)\in {\Omega_\hbar^\alpha}} \cdots  \rchi_{(w_\ell-u_\ell)\in {\Omega_\hbar^\alpha}} \rchi_{(w_{\ell+1}-u_{\ell+1})\in ({\Omega_\hbar^\alpha})^c}\cdots\rchi_{(w_{k}-u_{k})\in ({\Omega_\hbar^\alpha})^c},
	\end{equation}
	where $\ell \in \{1,\dots, k\}$. Thus, to continue from \eqref{LemRk_proof_0}, we have
	\begin{align*}
	&\left| \dotsint  (\mathrm{d}q\mathrm{d}p)^{\otimes k}  \nabla_{\vec{q}_k}\Phi(q_1,p_1,\dots,q_k,p_k) \cdot \mathcal{R}_k \right|\\
	\leq& C \hbar^{\frac{1}{2}-\frac{3}{2}k}  \sum_{j=1}^k \max_{0\leq \ell \leq k } \dotsint (\mathrm{d}q\mathrm{d}w\mathrm{d}u)^{\otimes k} \prod_{n\neq j}^k  \left|f\left( \frac{w_n-q_n}{\sqrt{\hbar}}\right)  f\left( \frac{u_n-q_n}{\sqrt{\hbar}}\right) \right| \left|\nabla f\left( \frac{w_j-q_j}{\sqrt{\hbar}}\right) \right|\left| f\left( \frac{u_j-q_j}{\sqrt{\hbar}}\right) \right|\\
	&\quad\cdot\bigg|\dotsint (\mathrm{d}p)^{\otimes k}  \left( \rchi_{(w_1-u_1)\in {\Omega_\hbar^\alpha}} \cdots  \rchi_{(w_\ell-u_\ell)\in {\Omega_\hbar^\alpha}} \rchi_{(w_{\ell+1}-u_{\ell+1})\in ({\Omega_\hbar^\alpha})^c}\cdots\rchi_{(w_{k}-u_{k})\in ({\Omega_\hbar^\alpha})^c} \right)  \nabla_{q_j} \Phi\cdot e^{\frac{i}{\hbar}\vec{p}_k \cdot (\vec{w}_k -\vec{u}_k)}\bigg| \\
	&\quad\cdot \norm{a_{w_k}\cdots a_{w_1} \psi_{N,t}}\norm{ a_{u_k}\cdots a_{u_1} \psi_{N,t}}\\
	\leq & C \hbar^{\frac{1}{2}-\frac{3}{2}k}  \sum_{j=1}^k \max_{0\leq \ell \leq k }  \dotsint (\mathrm{d}q\mathrm{d}w\mathrm{d}u)^{\otimes k} \prod_{n\neq j}^k  \left|f\left( \frac{w_n-q_n}{\sqrt{\hbar}}\right)  f\left( \frac{u_n-q_n}{\sqrt{\hbar}}\right) \right| \left|\nabla f\left( \frac{w_j-q_j}{\sqrt{\hbar}}\right) \right|\left| f\left( \frac{u_j-q_j}{\sqrt{\hbar}}\right) \right| \\
	& \quad\cdot\bigg| \dotsint (\mathrm{d}p)^{\otimes \ell} \rchi_{(w_1-u_1)\in {\Omega_\hbar^\alpha}} \cdots  \rchi_{(w_\ell-u_\ell)\in {\Omega_\hbar^\alpha}} e^{\frac{i}{\hbar} \sum_{m=1}^\ell {p}_{m} \cdot ({w}_{m} -{u}_{m})} \\
	&\qquad\cdot \dotsint (\mathrm{d}p)^{\otimes (k-\ell)}  \rchi_{(w_{\ell+1}-u_{\ell+1})\in ({\Omega_\hbar^\alpha})^c}\cdots\rchi_{(w_{k}-u_{k})\in ({\Omega_\hbar^\alpha})^c} e^{\frac{i}{\hbar} \sum_{m=k-\ell}^\ell {p}_{m} \cdot ({w}_{m} -{u}_{m})} \nabla_{q_j} \Phi \bigg| \\
	&\quad\cdot \norm{a_{w_k}\cdots a_{w_1} \psi_{N,t}}\norm{ a_{u_k}\cdots a_{u_1} \psi_{N,t}}
	\end{align*}
	Applying Lemma \ref{estimate_oscillation} onto the $(k-\ell)$ terms, we have
	\begin{align*}
	\leq &  C  \sum_{j=1}^k \max_{0\leq \ell \leq k } \hbar^{\frac{1}{2}-\frac{3}{2}k + (1-\alpha)(k-\ell)s}  \dotsint (\mathrm{d}q\mathrm{d}w\mathrm{d}u)^{\otimes k}  \left(\rchi_{(w_1-u_1)\in {\Omega_\hbar^\alpha}} \cdots  \rchi_{(w_\ell-u_\ell)\in {\Omega_\hbar^\alpha}} \right)\\
	&\qquad\cdot \prod_{n\neq j}^k \left|f\left( \frac{w_n-q_n}{\sqrt{\hbar}}\right)  f\left( \frac{u_n-q_n}{\sqrt{\hbar}}\right) \right|\cdot \left|\nabla f\left( \frac{w_j-q_j}{\sqrt{\hbar}}\right) \right|\left| f\left( \frac{u_j-q_j}{\sqrt{\hbar}}\right) \right|\\
	&\qquad\cdot \norm{a_{w_k}\cdots a_{w_1} \psi_{N,t}}\norm{ a_{u_k}\cdots a_{u_1} \psi_{N,t}}.\\
	\end{align*}
	For a fixed $\ell$, observe that since $f$ is compact supported, by using H\"older's inequality in $w$ and $u$ variables, we have
	\begin{align*}
	& \dotsint (\mathrm{d}q\mathrm{d}w\mathrm{d}u)^{\otimes k}  \left(\rchi_{(w_1-u_1)\in {\Omega_\hbar^\alpha}} \cdots  \rchi_{(w_\ell-u_\ell)\in {\Omega_\hbar^\alpha}} \right) \prod_{n\neq j}^k \left|f\left( \frac{w_n-q_n}{\sqrt{\hbar}}\right)  f\left( \frac{u_n-q_n}{\sqrt{\hbar}}\right) \right|\\
	&\quad\cdot\left|\nabla f\left( \frac{w_j-q_j}{\sqrt{\hbar}}\right) \right|\left| f\left( \frac{u_j-q_j}{\sqrt{\hbar}}\right) \right| \norm{a_{w_k}\cdots a_{w_1} \psi_{N,t}}\norm{ a_{u_k}\cdots a_{u_1} \psi_{N,t}}\\
	=& \dotsint (\mathrm{d}q\mathrm{d}w\mathrm{d}u)^{\otimes k}   \left(\rchi_{(w_1-u_1)\in {\Omega_\hbar^\alpha}} \cdots  \rchi_{(w_\ell-u_\ell)\in {\Omega_\hbar^\alpha}} \right) \prod_{n\neq j}^k \left|f\left( \frac{w_n-q_n}{\sqrt{\hbar}}\right)  f\left( \frac{u_n-q_n}{\sqrt{\hbar}}\right) \right|\\
	&\quad\cdot \left|\nabla f\left( \frac{w_j-q_j}{\sqrt{\hbar}}\right) \right|\left| f\left( \frac{u_j-q_j}{\sqrt{\hbar}}\right)\right| \prod_{m=1}^k \rchi_{|w_m - q_m| \leq \sqrt{\hbar} R} \rchi_{|u_m - q_m| \leq \sqrt{\hbar} R} \norm{a_{w_k}\cdots a_{w_1} \psi_{N,t}}\norm{ a_{u_k}\cdots a_{u_1} \psi_{N,t}}
	\end{align*}
Since $2\norm{a_{w_k}\cdots a_{w_1} \psi_{N,t}}\norm{ a_{u_k}\cdots a_{u_1} \psi_{N,t}}\leq \norm{a_{w_k}\cdots a_{w_1} \psi_{N,t}}^2+\norm{ a_{u_k}\cdots a_{u_1} \psi_{N,t}}^2$, we only need to estimate the term with one of these. Namely,	
\begin{align*}	
	&\dotsint (\mathrm{d}q\mathrm{d}w\mathrm{d}u)^{\otimes k}   \left(\rchi_{(w_1-u_1)\in {\Omega_\hbar^\alpha}} \cdots  \rchi_{(w_\ell-u_\ell)\in {\Omega_\hbar^\alpha}} \right) \prod_{n\neq j}^k \left|f\left( \frac{w_n-q_n}{\sqrt{\hbar}}\right)  f\left( \frac{u_n-q_n}{\sqrt{\hbar}}\right) \right|\\
	&\quad\cdot \left|\nabla f\left( \frac{w_j-q_j}{\sqrt{\hbar}}\right) \right|\left| f\left( \frac{u_j-q_j}{\sqrt{\hbar}}\right)\right|  \norm{a_{w_k}\cdots a_{w_1} \psi_{N,t}}^2\\
	\leq & \norm{f}_{L^\infty}^k \hbar^{3\alpha l}\dotsint (\mathrm{d}q\mathrm{d}w)^{\otimes k}    \prod_{n\neq j}^k \left|f\left( \frac{w_n-q_n}{\sqrt{\hbar}}\right) \right|\cdot \left|\nabla f\left( \frac{w_j-q_j}{\sqrt{\hbar}}\right) \right| \norm{a_{w_k}\cdots a_{w_1} \psi_{N,t}}^2\\
	\leq & C \hbar^{3\alpha l+\frac{3}{2}k-3k}=C \hbar^{3\alpha l-\frac{3}{2}k}.
	\end{align*}
	Then, from \eqref{LemRk_proof_1}, we have
	\begin{equation}\label{LemRk_proof_3}
	\begin{aligned}
	\left| \dotsint  (\mathrm{d}q\mathrm{d}p)^{\otimes k}  \nabla_{\vec{q}_k}\Phi(q_1,p_1,\dots,q_k,p_k) \cdot \mathcal{R}_k \right|
	&\leq  C  \sum_{j=1}^k  \max_{0\leq \ell \leq k } \hbar^{\frac{1}{2}-\frac{3}{2}k + (1-\alpha)(k-\ell)s+{3\alpha \ell - \frac{3}{2}k}} \\
	& = C k  \max_{0\leq \ell \leq k }
	\hbar^{\frac{1}{2}{-3k + (1-\alpha)sk+(3\alpha-(1-\alpha) s) \ell}}.
	\end{aligned}
	\end{equation}
	Therefore, by picking $s = {\left\lceil \frac{3\alpha}{1-\alpha} \right\rceil}$ we arrive immediately that
	\begin{align*}
	\left| \dotsint  (\mathrm{d}q\mathrm{d}p)^{\otimes k}  \nabla_{\vec{q}_k}\Phi(q_1,p_1,\dots,q_k,p_k) \cdot \mathcal{R}_k \right|\leq {C \hbar^{\frac{1}{2}-3k+3\alpha k}=C \hbar^{\frac{1}{2}-3k(1-\alpha)}}.
	\end{align*}
	Therefore, for all $\delta\ll 1$, we choose $\frac{1}{2}<\alpha <1$ such that ${\hbar^{-3k(1-\alpha)}\leq \hbar^{- \delta}}$. 
	
\end{proof}

\subsubsection{Proof of Proposition \ref{lemma_vla_k1_convg}}
\begin{proof}
	Let $\Phi$ be an arbitrary test function, then the remainder term $\widetilde{\mathcal{R}}_1$ can be written explicitly into
	\begin{align*}
	& \bigg|  \iint \mathrm{d}q_1 \mathrm{d}p_1 \nabla_{p_1} \Phi(q_1, p_1) \cdot \widetilde{\mathcal{R}}_1\bigg|\\
	=&\bigg|  \iint \mathrm{d}q_1 \mathrm{d}p_1 \nabla_{p_1} \Phi(q_1, p_1) \cdot \bigg(\iint \mathrm{d}w\mathrm{d}u \iint \mathrm{d}y\mathrm{d}v \iint \mathrm{d}q_2 \mathrm{d}p_2 \\
	&\qquad\cdot\Big[\int_0^1 \mathrm{d}s  \nabla V\big(su+(1-s)w - y \big)- \nabla V(q_1-q_2)\Big]\\
	&\qquad\cdot f_{q_1,p_1}^\hbar (w) \overline{f_{q_1,p_1}^\hbar (u)}  f_{q_2,p_2}^\hbar (y) \overline{f_{q_2,p_2}^\hbar (v)} \left< a_w a_y \psi_{N,t}, a_u a_v \psi_{N,t} \right> \bigg) \bigg|\\
	=& \frac{1}{\hbar^3}  \bigg|\iint \mathrm{d}q_1 \mathrm{d}p_1 \nabla_{p_1} \Phi(q_1, p_1) \cdot \iint \mathrm{d}w\mathrm{d}u \iint \mathrm{d}y\mathrm{d}v \iint \mathrm{d}q_2 \mathrm{d}p_2  \\
	&\qquad\cdot\Big[\int_0^1 \mathrm{d}s  \nabla V\big(su+(1-s)w - y \big)- \nabla V(q_1-q_2)\Big]e^{\frac{\mathrm{i}}{\hbar}p_1\cdot(w-u)}e^{\frac{\mathrm{i}}{\hbar}p_2\cdot(y-v)}\\
	& \qquad \cdot f \left( \frac{w-q_1}{\sqrt{\hbar}} \right) \overline{f\left(\frac{u-q_1}{\sqrt{\hbar}}\right)} f \left( \frac{y-q_2}{\sqrt{\hbar}} \right) \overline{f\left(\frac{v-q_2}{\sqrt{\hbar}}\right)}  \left< a_w a_y \psi_{N,t}, a_u a_v \psi_{N,t} \right> \bigg|.
	\end{align*}
	Then, utilizing \eqref{hbar_fourier}, we may get
	\begin{align*}
	&(2\pi)^3  \bigg|\iint \mathrm{d}q_1 \mathrm{d}p_1 \nabla_{p_1} \Phi(q_1, p_1) \cdot  \iint \mathrm{d}w\mathrm{d}u \iint \mathrm{d}y \mathrm{d}q_2 \\
	&\qquad \cdot \left[ \int_0^1 \mathrm{d}s \nabla V (su+(1-s)w-y) - \nabla V(q_1-q_2) \right] \\
	&\qquad\cdot f \left( \frac{w-q_1}{\sqrt{\hbar}} \right) \overline{f\left(\frac{u-q_1}{\sqrt{\hbar}}\right)} e^{\frac{\mathrm{i}}{\hbar}p_1\cdot(w-u)} \left| f \left( \frac{y-q_2}{\sqrt{\hbar}} \right)\right|^2 \left< a_w a_y \psi_{N,t}, a_u a_y \psi_{N,t} \right> \bigg|\\
	= &  (2\pi)^{3} \hbar^{\frac{3}{2}} \bigg| \iint \mathrm{d}q_1 \mathrm{d}p_1 \nabla_{p_1} \Phi(q_1, p_1) \cdot \iint \mathrm{d}w\mathrm{d}u \iint \mathrm{d}y \mathrm{d}\widetilde{q}_2 \\
	&\qquad \cdot \left[ \int_0^1 \mathrm{d}s \nabla V (su+(1-s)w-y) - \nabla V(q_1-y +\sqrt{\hbar}\widetilde{q}_2) \right] \\
	&\qquad\cdot f \left( \frac{w-q_1}{\sqrt{\hbar}} \right) \overline{f\left(\frac{u-q_1}{\sqrt{\hbar}}\right)} e^{\frac{\mathrm{i}}{\hbar}p_1\cdot(w-u)} \left| f \left( \widetilde{q}_2 \right)\right|^2  \left< a_w a_y \psi_{N,t}, a_u a_y \psi_{N,t} \right> \bigg|.
	\end{align*}
	Then, we insert a term, namely $\nabla V(q_1-y)$ and use triangle inequality to obtain
	\begin{align*}
	\leq &  (2\pi)^3\hbar^{\frac{3}{2}}  \bigg| \iint \mathrm{d}q_1 \mathrm{d}p_1 \nabla_{p_1} \Phi(q_1, p_1) \cdot \iint \mathrm{d}w\mathrm{d}u \iint \mathrm{d}y \mathrm{d}\widetilde{q}_2 \\
	&\qquad\cdot \int_0^1 \mathrm{d}s\bigg( \nabla V (su+(1-s)w-y) - \nabla V(q_1-y) \bigg)\\
	&\qquad\cdot f \left( \frac{w-q_1}{\sqrt{\hbar}} \right) \overline{f\left(\frac{u-q_1}{\sqrt{\hbar}}\right)} e^{\frac{\mathrm{i}}{\hbar}p_1\cdot(w-u)} \big| f \left(\widetilde{q}_2 \right)\big|^2 \left< a_w a_y \psi_{N,t}, a_u a_y \psi_{N,t} \right> \bigg| \\
	&+ (2\pi)^3\hbar^{\frac{3}{2}}  \bigg| \iint \mathrm{d}q_1 \mathrm{d}p_1 \nabla_{p_1} \Phi(q_1, p_1) \cdot \iint \mathrm{d}w\mathrm{d}u \iint \mathrm{d}y \mathrm{d}\widetilde{q}_2 \bigg(\nabla V(q_1-y) - \nabla V(q_1-y+\sqrt{\hbar}\widetilde{q}_2)\bigg)\\
	&\qquad\cdot f \left( \frac{w-q_1}{\sqrt{\hbar}} \right)  \overline{f\left(\frac{u-q_1}{\sqrt{\hbar}}\right)}e^{\frac{\mathrm{i}}{\hbar}p_1\cdot(w-u)} \big| f \left( \widetilde{q}_2 \right)\big|^2\left< a_w a_y \psi_{N,t}, a_u a_y \psi_{N,t} \right> \bigg| \\
	=:& I_3 + \mathit{II}_3,
	\end{align*}
	where we have used change of variable $\sqrt{\hbar}\widetilde{q}_2 = (y-q_2)$ in the second term above. 
	
	We first focus on $\mathit{II}_3$. We begin by splitting the integral on momentum, by using Lemma \ref{estimate_oscillation}, it follows
	\begin{align} 
	\mathit{II}_3 
	=&  (2\pi)^3\hbar^{\frac{3}{2}} \bigg| \iint \mathrm{d}q_1 \mathrm{d}p_1 \nabla_{p_1} \Phi(q_1, p_1) \cdot \iint \mathrm{d}w\mathrm{d}u \iint \mathrm{d}y \mathrm{d}\widetilde{q}_2 \left(\rchi_{(w-u)\in ({\Omega_\hbar^\alpha})^c} +\rchi_{(w-u)\in {\Omega_\hbar^\alpha}}\right)\notag\\
	&\qquad\cdot  \bigg( \nabla V(q_1-y) - \nabla V(q_1-y+\sqrt{\hbar}\widetilde{q}_2) \bigg) f \left( \frac{w-q_1}{\sqrt{\hbar}} \right) \overline{f\left(\frac{u-q_1}{\sqrt{\hbar}}\right)}\notag\\
	&\qquad\cdot e^{\frac{\mathrm{i}}{\hbar}p_1\cdot(w-u)} \big| f \left( \widetilde{q}_2 \right)\big|^2 \left< a_w a_y \psi_{N,t}, a_u a_y \psi_{N,t} \right> \bigg|\notag\\
	\leq&
	(2\pi)^3\hbar^{\frac{3}{2}+\frac{1}{2}}  \int \mathrm{d}q_1 \iint \mathrm{d}w\mathrm{d}u \int \mathrm{d}y \bigg(\left|\int \mathrm{d}p_1\ e^{\frac{\mathrm{i}}{\hbar}p_1\cdot(w-u)} \rchi_{(w-u)\in ({\Omega_\hbar^\alpha})^c} \nabla_{p_1} \Phi(q_1, p_1)  \right|\notag\\
	&\qquad +\left|\int \mathrm{d}p_1\ e^{\frac{\mathrm{i}}{\hbar}p_1\cdot(w-u)} \rchi_{(w-u)\in {\Omega_\hbar^\alpha}} \nabla_{p_1} \Phi(q_1, p_1)  \right| \bigg)\cdot \bigg| f \left( \frac{w-q_1}{\sqrt{\hbar}} \right) \overline{f\left(\frac{u-q_1}{\sqrt{\hbar}}\right)}\bigg|\notag \\
	&\qquad \cdot\left( \int \mathrm{d}\widetilde{q}_2 |\widetilde{q}_2| \big| f \left( \widetilde{q}_2 \right) 
	\big|^2\right) |\left< a_w a_y \psi_{N,t}, a_u a_y \psi_{N,t} \right>|\notag\\
	\leq&
	C \hbar^{\frac{3}{2}+\frac{1}{2}}  \int \mathrm{d}q_1 \iint \mathrm{d}w\mathrm{d}u \int \mathrm{d}y \bigg(\left|\int \mathrm{d}p_1\ e^{\frac{\mathrm{i}}{\hbar}p_1\cdot(w-u)} \rchi_{(w-u)\in ({\Omega_\hbar^\alpha})^c} \nabla_{p_1} \Phi(q_1, p_1)  \right|\notag\\
	&\qquad +\left|\int \mathrm{d}p_1\ e^{\frac{\mathrm{i}}{\hbar}p_1\cdot(w-u)} \rchi_{(w-u)\in {\Omega_\hbar^\alpha}} \nabla_{p_1} \Phi(q_1, p_1)  \right| \bigg)\cdot \bigg| f \left( \frac{w-q_1}{\sqrt{\hbar}} \right) \overline{f\left(\frac{u-q_1}{\sqrt{\hbar}}\right)}\bigg|\notag\\
	&\qquad\cdot|\left< a_w a_y \psi_{N,t}, a_u a_y \psi_{N,t} \right>|\notag\\
	=:& i_{31} + \mathit{ii}_{31},\label{proof_II_3}
	\end{align}
	where we used the fact that $\nabla V$ is Lipschitz continuous, $f$ has compact support, and the definition of ${\Omega_\hbar^\alpha}$ in \eqref{estimate_oscillation_omega}.
	
	The next step is to use Lemmata \ref{N_hbar} and \ref{estimate_oscillation} to bound the terms $i_{31}$ and $\mathit{ii}_{31}$. Then we examine what the appropriate terms $\alpha$ and $s$ should be. By Lemma \ref{estimate_oscillation}, we may bound the term $i_{31}$, i.e.,
	\begin{align*}
	i_{31} &\leq  C \hbar^{\frac{3}{2} +\frac{1}{2}+ (1-\alpha)s}  \int \mathrm{d}q_1 \iint \mathrm{d}w\mathrm{d}u \int \mathrm{d}y \cdot \bigg| f \left( \frac{w-q_1}{\sqrt{\hbar}} \right) \overline{f\left(\frac{u-q_1}{\sqrt{\hbar}}\right)}\bigg||\left< a_w a_y \psi_{N,t}, a_u a_y \psi_{N,t} \right>|\\	
	&
	\leq C \hbar^{\frac{3}{2} +\frac{1}{2} + (1-\alpha)s} \int \mathrm{d}q_1 \iint \mathrm{d}w \mathrm{d}u \int \mathrm{d}y \bigg|f \left( \frac{w-q_1}{\sqrt{\hbar}} \right) \overline{f\left(\frac{u-q_1}{\sqrt{\hbar}}\right)}\bigg| \norm{ a_w a_y \psi_{N,t}} \norm{ a_u a_y \psi_{N,t} }.
	\end{align*}
	Since we assume that $f$ is compactly supported, by H\"older inequality with respect to $w$ and $u$, we have we have that
	\begin{align*}
	i_{31}
	\leq &  C \hbar^{\frac{3}{2} +\frac{1}{2} + (1-\alpha)s}   \int \mathrm{d}q_1 \bigg( \iint \mathrm{d}w\mathrm{d}u \left|f\left(\frac{w-q_1}{\sqrt{\hbar}} \right) \right|^2 \left|f\left(\frac{u-q_1}{\sqrt{\hbar}} \right) \right|^2 \bigg)^\frac{1}{2}\\
	&\qquad\cdot \left( \iint \mathrm{d}w\mathrm{d}u\ \rchi_{|w-q_1|\leq \sqrt{\hbar}R}  \rchi_{|u-q_1|\leq \sqrt{\hbar}R} \left( \int \mathrm{d}y\ \norm{a_w a_y \psi_{N,t}} \norm{a_u a_y \psi_{N,t}} \right)^2 \right)^\frac{1}{2}\\
	= & C \hbar^{3 +\frac{1}{2} + (1-\alpha)s}  \int \mathrm{d}q_1  \bigg( \int \mathrm{d}\widetilde{w}  \left|f\left(\widetilde{w}\right) \right|^2  \bigg)\\
	&\qquad\cdot  \left( \iint \mathrm{d}w\mathrm{d}u\ \rchi_{|w-q_1|\leq \sqrt{\hbar}R} \rchi_{|u-q_1|\leq \sqrt{\hbar}R} \left( \int \mathrm{d}y\ \norm{a_w a_y \psi_{N,t}} \norm{a_u a_y \psi_{N,t}} \right)^2 \right)^\frac{1}{2},
	\end{align*}
	where we used the change of variable $\sqrt{\hbar}\widetilde{w} = w - q_1$ in the last inequality. Now, since $\norm{f}_2$ is normalized, we continue to have
	\begin{align*}
	&
	\leq  C \hbar^{3 +\frac{1}{2} + (1-\alpha)s} \int \mathrm{d}q_1 \left( \iint \mathrm{d}w\mathrm{d}u\ \rchi_{|w-q_1|\leq \sqrt{\hbar}R}  \rchi_{|u-q_1|\leq \sqrt{\hbar}R} \left( \int \mathrm{d}y\ \norm{a_w a_y \psi_{N,t}} \norm{a_u a_y \psi_{N,t}} \right)^2 \right)^\frac{1}{2}\\
	&
	\leq  C \hbar^{3 +\frac{1}{2} + (1-\alpha)s} \int \mathrm{d}q_1  \left( \iint \mathrm{d}w\mathrm{d}u\ \rchi_{|w-q_1|\leq \sqrt{\hbar}R}  \rchi_{|u-q_1|\leq\sqrt{\hbar}R} \left( \int \mathrm{d}y\ \norm{a_w a_y \psi_{N,t}}^2\right)\left( \int \mathrm{d}y\ \norm{a_u a_y \psi_{N,t}}^2 \right) \right)^\frac{1}{2} \\
	&
	=   C \hbar^{3 +\frac{1}{2} + (1-\alpha)s}\int \mathrm{d}q_1  \iint \mathrm{d}y \mathrm{d}w\ \rchi_{|w-q_1|\leq \sqrt{\hbar}R}\norm{a_w a_y \psi_{N,t}}^2\\
	&
	=  C \hbar^{3 +\frac{1}{2} + (1-\alpha)s}  \int \mathrm{d}y  \iint \mathrm{d}q_1 \mathrm{d}w \left<a_y \psi_{N,t}, \rchi_{|w-q_1|\leq \sqrt{\hbar}R} a^*_w a_w a_y \psi_{N,t} \right>
	\end{align*}
	by Lemma \ref{N_hbar}
	\begin{align}
	i_{31} 
	&\leq C \hbar^{3 - \frac{3}{2} +\frac{1}{2} + (1-\alpha)s}   \int \mathrm{d}y \left<a_y \psi_{N,t},  a_y \psi_{N,t} \right>\notag\\
	&
	=  C  \hbar^{(1-\alpha)s-1} \left<\psi_{N,t},\frac{\mathcal{N}}{N} \psi_{N,t} \right> \leq  C  \hbar^{(1-\alpha)s-1}.\label{i_31_0}
	\end{align}
	
	On the other hand, from $\mathit{ii}_{31}$ we have
	\begin{align*}
	\mathit{ii}_{31}  \leq & C \hbar^{\frac{3}{2}+\frac{1}{2}}  \int \mathrm{d}q_1 \iint \mathrm{d}w\mathrm{d}u \int \mathrm{d}y \left|\int \mathrm{d}p_1\ e^{\frac{\mathrm{i}}{\hbar}p_1\cdot(w-u)} \rchi_{(w-u)\in {\Omega_\hbar^\alpha}} \nabla_{p_1} \Phi(q_1, p_1)  \right| \\
	&\qquad\cdot\bigg| f \left( \frac{w-q_1}{\sqrt{\hbar}} \right) \overline{f\left(\frac{u-q_1}{\sqrt{\hbar}}\right)}\bigg||\left< a_w a_y \psi_{N,t}, a_u a_y \psi_{N,t} \right>|\\
	\leq & C \hbar^{\frac{3}{2}+\frac{1}{2}}  \int \mathrm{d}q_1 \iint \mathrm{d}w\mathrm{d}u \int \mathrm{d}y\int \mathrm{d}p_1\  \rchi_{(w-u)\in {\Omega_\hbar^\alpha}}  \left| \nabla_{p_1} \Phi(q_1, p_1)  \right| \\
	&\qquad\cdot \bigg| f \left( \frac{w-q_1}{\sqrt{\hbar}} \right) \overline{f\left(\frac{u-q_1}{\sqrt{\hbar}}\right)}\bigg||\left< a_w a_y \psi_{N,t}, a_u a_y \psi_{N,t} \right>|\\
	\leq & C \hbar^{\frac{3}{2}+\frac{1}{2}}  \int \mathrm{d}q_1 \iint \mathrm{d}w\mathrm{d}u \int \mathrm{d}y\  \rchi_{(w-u)\in {\Omega_\hbar^\alpha}} \cdot \bigg| f \left( \frac{w-q_1}{\sqrt{\hbar}} \right) \overline{f\left(\frac{u-q_1}{\sqrt{\hbar}}\right)}\bigg| |\left< a_w a_y \psi_{N,t}, a_u a_y \psi_{N,t} \right>|
	\end{align*}
	Since $f$ is assumed to be compactly supported, we have
	\begin{equation*}
	\begin{aligned}
	\leq & C \hbar^{\frac{3}{2}+\frac{1}{2}} 
	\int \mathrm{d}q_1 \iint  \mathrm{d}w\mathrm{d}u 
	\rchi_{(w-u)\in {\Omega_\hbar^\alpha}} 
	\left|f \left(\frac{w-q_1}{\sqrt{\hbar}}\right) 
	f \left(\frac{u-q_1}{\sqrt{\hbar}}\right) \right|
	\\
	&\qquad \left(\int \mathrm{d}y \norm{ a_w a_y \psi_{N,t}}^2\right)^\frac{1}{2} \cdot \left(\int \mathrm{d}y \norm{ a_u a_y \psi_{N,t}}^2\right)^\frac{1}{2}\\
	\leq & C \hbar^{\frac{3}{2}+\frac{1}{2}}\int \mathrm{d}q_1 \iint  \mathrm{d}w\mathrm{d}u 
	\rchi_{(w-u)\in {\Omega_\hbar^\alpha}} 
	\left|f \left(\frac{w-q_1}{\sqrt{\hbar}}\right) 
	f \left(\frac{u-q_1}{\sqrt{\hbar}}\right) \right|	
	 \rchi_{|w-q_1|\leq\sqrt{\hbar}R}\int \mathrm{d}y \norm{ a_w a_y \psi_{N,t}}^2\\
	 \leq & C \hbar^{\frac{3}{2}+\frac{1}{2}}\norm{f}_{L^\infty}\hbar^{3\alpha}\int \mathrm{d}q_1\int \mathrm{d}w \left|f \left(\frac{w-q_1}{\sqrt{\hbar}}\right) 
	 \right| \rchi_{|w-q_1|\leq\sqrt{\hbar}R}\norm{ a_w \mathcal{N}^{\frac{1}{2}}\psi_{N,t}}^2\\
	 \leq &	C\hbar^{\frac{3}{2}+\frac{1}{2}+3\alpha+\frac{3}{2}-6}= C\hbar^{-1+3(\alpha-\frac{1}{2})}.
	\end{aligned}
	\end{equation*}
where we have used the fact that the test function has compact support in the $q$ variable.

	Now we compare power of $\hbar$ with the one in \eqref{i_31_0} and choose 
	\begin{eqnarray*}\label{s}
	s = \left\lceil \frac{3(\alpha-\frac{1}{2})}{1-\alpha} \right\rceil
	\end{eqnarray*}
	 such that $\mathit{II}_3$ is of order $\hbar^{\frac{1}{2}+3(\alpha-1)}$.

	Now, focus on $I_3$, we use similar strategy as with $\mathit{II}_3$.
	\begin{align}
	I_3  \leq & C \hbar^{\frac{3}{2}}  \int \mathrm{d}q_1 \iint \mathrm{d}w\mathrm{d}u \int \mathrm{d}y  \int_0^1 \mathrm{d}s\bigg| \nabla V (su+(1-s)w-y) - \nabla V(q_1-y) \bigg|\notag\\
	&\qquad\cdot \bigg(\left|\int \mathrm{d}p_1\ e^{\frac{\mathrm{i}}{\hbar}p_1\cdot(w-u)} \rchi_{(w-u)\in ({\Omega_\hbar^\alpha})^c} \nabla_{p_1} \Phi(q_1, p_1)  \right| +\left|\int \mathrm{d}p_1\ e^{\frac{\mathrm{i}}{\hbar}p_1\cdot(w-u)} \rchi_{(w-u)\in {\Omega_\hbar^\alpha}} \nabla_{p_1} \Phi(q_1, p_1)  \right| \bigg)\notag\\
	&\qquad\cdot \bigg| f \left( \frac{w-q_1}{\sqrt{\hbar}} \right) \overline{f\left(\frac{u-q_1}{\sqrt{\hbar}}\right)}\bigg|\left(\int \mathrm{d}\widetilde{q}_2 |f(\widetilde{q}_2)|^2 \right)|\left< a_w a_y \psi_{N,t}, a_u a_y \psi_{N,t} \right>|\notag\\
	\leq &  C \hbar^{\frac{3}{2}}  \int \mathrm{d}q_1 \iint \mathrm{d}w\mathrm{d}u \int \mathrm{d}y  \int_0^1 \mathrm{d}s |su+(1-s)w-q_1|
	\bigg(\left|\int \mathrm{d}p_1\ e^{\frac{\mathrm{i}}{\hbar}p_1\cdot(w-u)} \rchi_{(w-u)\in ({\Omega_\hbar^\alpha})^c} \nabla_{p_1} \Phi(q_1, p_1)  \right|\notag\\
	&\qquad+\left|\int \mathrm{d}p_1\ e^{\frac{\mathrm{i}}{\hbar}p_1\cdot(w-u)}  \rchi_{(w-u)\in {\Omega_\hbar^\alpha}} \nabla_{p_1} \Phi(q_1, p_1)  \right| \bigg)\cdot \bigg| f \left( \frac{w-q_1}{\sqrt{\hbar}} \right) \overline{f\left(\frac{u-q_1}{\sqrt{\hbar}}\right)}\bigg|\notag\\
	&\qquad\cdot \rchi_{|w-q_1|\leq \sqrt{\hbar}R} \rchi_{|u-q_1|\leq \sqrt{\hbar}R}\norm{a_w a_y \psi_{N,t}} \norm{a_u a_y \psi_{N,t}}\notag\\
	&
	=:  i_{32}+\mathit{ii}_{32}  \label{proof_I_3}
	\end{align}
	Again, by Lemma \ref{estimate_oscillation} and the bounds for number operator and localized number operator, we have for $i_{32}$ that
	by using the symmetric property of the integrals
	\begin{align*}
	i_{32}  \leq & C \hbar^{\frac{3}{2}+ (1-\alpha)s}  \int \mathrm{d}q_1 \iint \mathrm{d}w\mathrm{d}u \int_0^1 \mathrm{d}s\  |su+(1-s)w-q_1| \cdot \bigg| f \left( \frac{w-q_1}{\sqrt{\hbar}} \right) \overline{f\left(\frac{u-q_1}{\sqrt{\hbar}}\right)}\bigg|\\
		& \qquad\cdot \rchi_{|w-q_1|\leq \sqrt{\hbar}R}\rchi_{|u-q_1|\leq \sqrt{\hbar}R}  \int \mathrm{d}y \norm{a_w a_y \psi_{N,t}} \norm{a_u a_y \psi_{N,t}}\\
	\leq & C \hbar^{\frac{3}{2}+ (1-\alpha)s}  \int \mathrm{d}q_1 \iint \mathrm{d}w\mathrm{d}u (|u-q_1|+|w-q_1|) \cdot \bigg| f \left( \frac{w-q_1}{\sqrt{\hbar}} \right) \overline{f\left(\frac{u-q_1}{\sqrt{\hbar}}\right)}\bigg|\\
	& \qquad\cdot \rchi_{|w-q_1|\leq \sqrt{\hbar}R}\rchi_{|u-q_1|\leq \sqrt{\hbar}R} \left( \int \mathrm{d}y \norm{a_w a_y \psi_{N,t}}^2\right)^{\frac{1}{2}} \left( \int \mathrm{d}y \norm{a_u a_y \psi_{N,t}}^2\right)^{\frac{1}{2}}\\
		\leq & C \hbar^{\frac{3}{2}+ (1-\alpha)s}  \int \mathrm{d}q_1 \iint \mathrm{d}w\mathrm{d}u |u-q_1| \cdot \bigg| f \left( \frac{w-q_1}{\sqrt{\hbar}} \right)\bigg| \bigg|f\left(\frac{u-q_1}{\sqrt{\hbar}}\right)\bigg|\\
	& \qquad\cdot \rchi_{|w-q_1|\leq \sqrt{\hbar}R}\rchi_{|u-q_1|\leq \sqrt{\hbar}R} \int \mathrm{d}y \norm{a_w a_y \psi_{N,t}}^2\\
	\leq & C \hbar^{3 +\frac{1}{2} + (1-\alpha)s} \int \mathrm{d}\tilde u |f(\tilde u)|^2 \int \mathrm{d}q_1\int \mathrm{d}w \left|f \left(\frac{w-q_1}{\sqrt{\hbar}}\right) 
	\right| \rchi_{|w-q_1|\leq\sqrt{\hbar}R}\norm{ a_w \mathcal{N}^{\frac{1}{2}}\psi_{N,t}}^2\\
	\leq & C \hbar^{3+ \frac{1}{2}+(1-\alpha)s+\frac{3}{2}-6}=C\hbar^{-1+(1-\alpha)s}.
	\end{align*}
	where we used Lemma \ref{N_hbar} and the bounds for number operator. Similarly, for $\mathit{ii}_{32}$, we have
	\begin{align*}
	\mathit{ii}_{32} \leq &  C \hbar^{\frac{3}{2}}  \int \mathrm{d}q_1 \iint \mathrm{d}w\mathrm{d}u \int \mathrm{d}y  \int_0^1 \mathrm{d}s |su+(1-s)w-q_1|\int \mathrm{d}p_1\ \left|\rchi_{(w-u)\in {\Omega_\hbar^\alpha}}  \nabla_{p_1} \Phi(q_1, p_1)  \right|\\
	&\qquad\cdot \bigg| f \left( \frac{w-q_1}{\sqrt{\hbar}} \right) \overline{f\left(\frac{u-q_1}{\sqrt{\hbar}}\right)}\bigg| \rchi_{|w-q_1|\leq \sqrt{\hbar}R} \rchi_{|u-q_1|\leq \sqrt{\hbar}R}\norm{a_w a_y \psi_{N,t}} \norm{a_u a_y \psi_{N,t}}\\
	\leq   &C \hbar^{\frac{3}{2}} \int \mathrm{d}q_1 \iint \mathrm{d}w\mathrm{d}u \int \mathrm{d}y  (|u-q_1|+|w-q_1|)\rchi_{(w-u)\in {\Omega_\hbar^\alpha}} \bigg| f \left( \frac{w-q_1}{\sqrt{\hbar}} \right) \overline{f\left(\frac{u-q_1}{\sqrt{\hbar}}\right)}\bigg|\\
	&\qquad\cdot \rchi_{|w-q_1|\leq \sqrt{\hbar}R} \rchi_{|u-q_1|\leq \sqrt{\hbar}R}\norm{a_w a_y \psi_{N,t}} \norm{a_u a_y \psi_{N,t}}\\
    \leq   &C \hbar^{\frac{3}{2}} \int \mathrm{d}q_1 \iint \mathrm{d}w\mathrm{d}u \int \mathrm{d}y  |u-q_1|\rchi_{(w-u)\in {\Omega_\hbar^\alpha}} \bigg| f \left( \frac{w-q_1}{\sqrt{\hbar}} \right)\bigg|
    \bigg|f\left(\frac{u-q_1}{\sqrt{\hbar}}\right)\bigg|\norm{a_u \mathcal{N}^{\frac{1}{2}} \psi_{N,t}}^2\\
    \leq & C \hbar^{\frac{3}{2}+\frac{1}{2}}\norm{f}_{L^\infty}\hbar^{3\alpha}\int \mathrm{d}q_1\int \mathrm{d}w \left|f \left(\frac{w-q_1}{\sqrt{\hbar}}\right) 
    \right| \rchi_{|w-q_1|\leq\sqrt{\hbar}R}\norm{ a_w \mathcal{N}^{\frac{1}{2}}\psi_{N,t}}^2\\
    \leq & C\hbar^{-1+3(\alpha-\frac{1}{2})}.
    \end{align*}
	Then, by choosing the same $s$ as in \eqref{s} we get the same estimate for $I_3$.
	Therefore, $\mathit{II}_3$ and $I_3$ together, we have the bound of order $\hbar^{\frac{1}{2}+3(\alpha-1)}$ for $\alpha\in(\frac{1}{2},1)$.
\end{proof}

\subsubsection{Proof of Proposition \ref{mk_BBFKY_to_infty}}

\begin{proof}	
	To calculate the bound in \eqref{mk_BBFKY_to_infty_1} for $\widehat{\mathcal R}_k$. It has automatically an $1/N$ as a factor, therefore, we expect it has better estimates than the other remainder terms. More precisely, we can split the integrals as before,
	\begin{align*}
	&\bigg| \frac{1}{2N} \dotsint (\mathrm{d}q \mathrm{d}p)^{\otimes k} (\mathrm{d}w\mathrm{d}u)^{\otimes k}  \Phi(q_1, p_1,\dots,q_k,p_k) \sum_{j\neq i}^k  \bigg[ V(u_j-u_i) - V(w_j -w_i) \bigg] \\
	&\qquad\cdot \left( f_{q,p}^\hbar (w) \overline{f_{q,p}^\hbar (u)} \right)^{\otimes k}   \left< a_{w_k} \cdots a_{w_1} \psi_{N,t}, a_{u_k} \cdots a_{u_1} \psi_{N,t} \right> \bigg| \\
	=  &\bigg| \frac{1}{2N} \dotsint (\mathrm{d}q \mathrm{d}p)^{\otimes k} (\mathrm{d}w\mathrm{d}u)^{\otimes k}  \Phi(q_1, p_1,\dots,q_k,p_k) \sum_{j\neq i}^k  \bigg[ V(u_j-u_i) - V(w_j -w_i) \bigg]  \left( f_{q,p}^\hbar (w) \overline{f_{q,p}^\hbar (u)} \right)^{\otimes k} \\
	&\qquad\cdot\prod_{n=1}^k \left(\rchi_{(w_n-u_n)\in ({\Omega_\hbar^\alpha})^c}+ \rchi_{(w_n-u_n)\in {\Omega_\hbar^\alpha}} \right) \left< a_{w_k} \cdots a_{w_1} \psi_{N,t}, a_{u_k} \cdots a_{u_1} \psi_{N,t} \right> \bigg|,
	\end{align*}
	where ${\Omega_\hbar^\alpha}$ is defined as in \eqref{estimate_oscillation_omega}. Since $V \in W^{2,\infty}$ and recall $\hbar^3 = N^{-1}$, we have
	\begin{align*}
	\leq &C(k)\norm{V}_\infty \hbar^{3-\frac{3}{2}k} \dotsint (\mathrm{d}q \mathrm{d}w\mathrm{d}u)^{\otimes k}\prod_{n=1}^k \left|f \left(\frac{w_n-q_n}{\sqrt{\hbar}}\right)f \left(\frac{u_n-q_n}{\sqrt{\hbar}}\right)\right|\norm{ a_{w_k} \cdots a_{w_1} \psi_{N,t}} \norm{a_{u_k} \cdots a_{u_1} \psi_{N,t} }\\
	&\qquad\cdot \Bigg|\prod_{n=1}^k \left(\rchi_{(w_n-u_n)\in ({\Omega_\hbar^\alpha})^c}+ \rchi_{(w_n-u_n)\in {\Omega_\hbar^\alpha}} \right) \dotsint (\mathrm{d}p)^{\otimes k} e^{\frac{i}{\hbar}\sum_{m=1}^k p_m \cdot (w_m -u_m)} \Phi(q_1,\dots, p_k) \Bigg| \\
	\leq & C \hbar^{3-\frac{3}{2}k}  \max_{0\leq \ell \leq k } \dotsint (\mathrm{d}q\mathrm{d}w\mathrm{d}u)^{\otimes k}\prod_{n=1}^k  \left|f\left( \frac{w_n-q_n}{\sqrt{\hbar}}\right)  f\left( \frac{u_n-q_n}{\sqrt{\hbar}}\right) \right|\norm{ a_{w_k} \cdots a_{w_1} \psi_{N,t}} \norm{a_{u_k} \cdots a_{u_1} \psi_{N,t} }\\
	&\qquad\cdot \bigg|\dotsint (\mathrm{d}p)^{\otimes k}  \left( \rchi_{(w_1-u_1)\in {\Omega_\hbar^\alpha}} \cdots  \rchi_{(w_\ell-u_\ell)\in {\Omega_\hbar^\alpha}} \rchi_{(w_{\ell+1}-u_{\ell+1})\in ({\Omega_\hbar^\alpha})^c}\cdots\rchi_{(w_{k}-u_{k})\in ({\Omega_\hbar^\alpha})^c} \right)  \nabla_{q_j} \Phi\cdot e^{\frac{i}{\hbar}\vec{p}_k \cdot (\vec{w}_k -\vec{u}_k)}\bigg| \\
	=&C \hbar^{3-\frac{3}{2}k} \max_{0\leq \ell \leq k }  \dotsint (\mathrm{d}q\mathrm{d}w\mathrm{d}u)^{\otimes k}  \bigg| \dotsint (\mathrm{d}p)^{\otimes \ell} \rchi_{(w_1-u_1)\in {\Omega_\hbar^\alpha}} \cdots  \rchi_{(w_\ell-u_\ell)\in {\Omega_\hbar^\alpha}} e^{\frac{i}{\hbar} \sum_{m=1}^\ell {p}_{m} \cdot ({w}_{m} -{u}_{m})} \\
	&\qquad\cdot\dotsint (\mathrm{d}p)^{\otimes (k-\ell)}  \rchi_{(w_{\ell+1}-u_{\ell+1})\in ({\Omega_\hbar^\alpha})^c}\cdots\rchi_{(w_{k}-u_{k})\in ({\Omega_\hbar^\alpha})^c} e^{\frac{i}{\hbar} \sum_{m=k-\ell}^\ell {p}_{m} \cdot ({w}_{m} -{u}_{m})} \nabla_{q_j} \Phi(q_1,p_1,\dots,q_k,p_k) \bigg| \\
	&\qquad\cdot\prod_{n=1}^k  \left|f\left( \frac{w_n-q_n}{\sqrt{\hbar}}\right)  f\left( \frac{u_n-q_n}{\sqrt{\hbar}}\right) \right|  \norm{a_{w_k}\cdots a_{w_1} \psi_{N,t}}\norm{ a_{u_k}\cdots a_{u_1} \psi_{N,t}},
	\end{align*}
	where we apply similar argument in \eqref{LemRk_proof_1} in the last inequality. Note here that the constant $C$ above is dependent on $k$. Applying Lemma \ref{estimate_oscillation} we have
	\begin{align*}
	\leq & C \max_{0\leq \ell \leq k }  \hbar^{3-\frac{3}{2}k+ (1-\alpha)(k-\ell)s} \dotsint (\mathrm{d}q\mathrm{d}w\mathrm{d}u)^{\otimes k}  \left(\rchi_{(w_1-u_1)\in {\Omega_\hbar^\alpha}} \cdots  \rchi_{(w_\ell-u_\ell)\in {\Omega_\hbar^\alpha}} \right) \\
	&\qquad\cdot\prod_{n = 1}^k \left|f\left( \frac{w_n-q_n}{\sqrt{\hbar}}\right)  f\left( \frac{u_n-q_n}{\sqrt{\hbar}}\right) \right|  \norm{a_{w_k}\cdots a_{w_1} \psi_{N,t}}\norm{ a_{u_k}\cdots a_{u_1} \psi_{N,t}} \\
	= & C \max_{0\leq \ell \leq k }  \hbar^{3-\frac{3}{2}k+ (1-\alpha)(k-\ell)s}  \dotsint (\mathrm{d}q\mathrm{d}w\mathrm{d}u)^{\otimes k}  \left(\rchi_{(w_1-u_1)\in {\Omega_\hbar^\alpha}} \cdots  \rchi_{(w_\ell-u_\ell)\in {\Omega_\hbar^\alpha}} \right) \\
	&\qquad\cdot\prod_{n = 1}^k \left|f\left( \frac{w_n-q_n}{\sqrt{\hbar}}\right)  f\left( \frac{u_n-q_n}{\sqrt{\hbar}}\right) \right| \rchi_{|w_n-q_n|\leq\sqrt{\hbar}R}\rchi_{|u_n-q_n|\leq\sqrt{\hbar}R}  \norm{a_{w_k}\cdots a_{w_1} \psi_{N,t}}\norm{ a_{u_k}\cdots a_{u_1} \psi_{N,t}}\\
	\leq &   C \max_{0\leq \ell \leq k }  \hbar^{3-\frac{3}{2}k+ (1-\alpha)(k-\ell)s}  \dotsint (\mathrm{d}q)^{\otimes k}  \dotsint (\mathrm{d}w\mathrm{d}u)^{\otimes k} \left(\rchi_{(w_1-u_1)\in {\Omega_\hbar^\alpha}} \cdots  \rchi_{(w_\ell-u_\ell)\in {\Omega_\hbar^\alpha}} \right)\\
	&\qquad\cdot \prod_{n = 1}^k \left|f\left( \frac{w_n-q_n}{\sqrt{\hbar}}\right)  f\left( \frac{u_n-q_n}{\sqrt{\hbar}}\right) \right|  \norm{a_{w_k}\cdots a_{w_1} \psi_{N,t}}^2 \\
	\leq  & C \max_{0\leq \ell \leq k }  \hbar^{3 -\frac{3}{2}k+ (1-\alpha)(k-\ell)s+3\alpha \ell}  \dotsint (\mathrm{d}q\mathrm{d}w)^{\otimes k}\prod_{n = 1}^k \left|f\left( \frac{w_n-q_n}{\sqrt{\hbar}}\right)\right| \norm{a_{w_k}\cdots a_{w_1} \psi_{N,t}}^2 \\
	\leq & C \max_{0\leq \ell \leq k }  \hbar^{3-\frac{3}{2}k+ (1-\alpha)(k-\ell)s +3\alpha \ell+\frac{3}{2}k-3k}\leq C \max_{0\leq \ell \leq k }  \hbar^{3-3k+(1-\alpha)sk+(3\alpha -(1-\alpha)s)\ell }.
	\end{align*}
	By choosing $s = \left\lceil \frac{3\alpha}{1-\alpha} \right\rceil$, we obtain estimate for $\widehat{\mathcal{R}}_k$.
	
	Next, we switch to estimate \eqref{mk_BBFKY_to_infty_2} for $\widetilde{\mathcal R}_k$. Repeated the steps in the proof of Proposition \ref{lemma_vla_k1_convg}, we have
	\begin{align*}
	&\left| \dotsint  (\mathrm{d}q \mathrm{d}p)^{\otimes k}\nabla_{\vec{p}_k} \Phi (q_1, p_1,\dots,q_k,p_k)  \cdot \widetilde{\mathcal{R}}_k \right|\\
	=&\bigg| \sum_{j=1}^k  \dotsint (\mathrm{d}q \mathrm{d}p)^{\otimes k} (\mathrm{d}w\mathrm{d}u)^{\otimes k} \nabla_{p_j} \Phi(q_1, p_1,\dots,q_k,p_k)  \cdot \iint \mathrm{d}y\mathrm{d}v \iint \mathrm{d}q_{k+1}\mathrm{d}p_{k+1}\\
	&\qquad\cdot \int_0^1 \mathrm{d}s \left[ \nabla V (su_j + (1-s)w_j - y) - \nabla V(q_j- y )+ \nabla V(q_j-y )- \nabla V(q_j- q_{k+1} ) \right] \\
	&\qquad\cdot\left( f^\hbar_{q,p} (w) \overline{f^\hbar_{q,p} (u)} \right)^{\otimes k} f^\hbar_{q_{k+1}, p_{k+1}} (y) \overline{f^\hbar_{q_{k+1}, p_{k+1}} (v)}  \left< a_{w_k} \cdots a_{w_1} a_y \psi_{N,t}, a_{u_k} \cdots a_{u_1} a_v \psi_{N,t} \right> \bigg|. \\
	\intertext{Appling the $\hbar$-weighted Dirac-delta function as in \eqref{hbar_fourier}, we have}
	=& (2\pi)^3 \hbar^{3-\frac{3}{2}} \bigg| \sum_{j=1}^k  \dotsint (\mathrm{d}q \mathrm{d}p)^{\otimes k} (\mathrm{d}w\mathrm{d}u)^{\otimes k} \nabla_{p_j} \Phi(q_1, p_1,\dots,q_k,p_k) \cdot \iint \mathrm{d}y\mathrm{d}q_{k+1}\\
	&\qquad\cdot \int_0^1 \mathrm{d}s \left[ \nabla V (su_j + (1-s)w_j - y) - \nabla V(q_j- y )+ \nabla V(q_j- y )- \nabla V(q_j- q_{k+1} ) \right] \\
	&\qquad\cdot\left( f^\hbar_{q,p} (w) \overline{f^\hbar_{q,p} (u)} \right)^{\otimes k} \left|f \left(\frac{y-q_{k+1}}{\sqrt{\hbar}} \right)\right|^2  \left< a_{w_k} \cdots a_{w_1} a_y \psi_{N,t}, a_{u_k} \cdots a_{u_1} a_y \psi_{N,t} \right> \bigg| \\
	\leq &  (2\pi)^3 \hbar^{3-\frac{3}{2}k} \sum_{j=1}^k \dotsint (\mathrm{d}q\mathrm{d}w\mathrm{d}u)^{\otimes k} \prod_{n=1}^k \left|\dotsint (\mathrm{d}p)^{\otimes k}  \nabla_{p_j}\Phi(q_1, p_1,\dots,q_k,p_k)  e^{\frac{\mathrm{i}}{\hbar} p_n \cdot (w_n-u_n)} \right|  \iint \mathrm{d}y\mathrm{d}\widetilde{q}_{k+1}\\
	&\qquad\cdot\left( \int_0^1 \mathrm{d}s \left| \nabla V (su_j + (1-s)w_j - y) - \nabla V(q_j- y ) ) \right|+\left|\nabla V(q_j-y) -\nabla V(q_j-y + \sqrt{\hbar}\widetilde{q}_{k+1} ) \right|\right)\\
	&\qquad\cdot\left|f \left(\frac{w_n-q_n}{\sqrt{\hbar}}\right)f \left(\frac{u_n-q_n}{\sqrt{\hbar}}\right)\right|\left|f \left(\widetilde{q}_{k+1} \right)\right|^2 |\left< a_{w_k} \cdots a_{w_1} a_y \psi_{N,t}, a_{u_k} \cdots a_{u_1} a_y \psi_{N,t} \right>|.\\
	\intertext{Using the fact that $\nabla V$ is Lipchitz continuous and that $f$ is compactly supported, we have}
	\leq & (2\pi)^3 \hbar^{3-\frac{3}{2}k} \sum_{j=1}^k \dotsint (\mathrm{d}q\mathrm{d}w\mathrm{d}u)^{\otimes k} \prod_{n=1}^k \left|\dotsint (\mathrm{d}p)^{\otimes k}  \nabla_{p_j}\Phi(q_1, p_1,\dots,q_k,p_k)  e^{\frac{\mathrm{i}}{\hbar} p_n \cdot (w_n-u_n)} \right|  \iint \mathrm{d}y\mathrm{d}\widetilde{q}_{k+1}\\
	&\qquad\cdot\left( \int_0^1 \mathrm{d}s \left|su_j + (1-s)w_j-q_j\right|+\left|\sqrt{\hbar}\widetilde{q}_{k+1} \right|\right)
	\left|f \left(\frac{w_n-q_n}{\sqrt{\hbar}}\right)f \left(\frac{u_n-q_n}{\sqrt{\hbar}}\right)\right| |f(\widetilde{q}_{k+1})|^2\\
	&\qquad\cdot \rchi_{|w_n-q_n|\leq\sqrt{\hbar}R}\rchi_{|u_n-q_n|\leq\sqrt{\hbar}R} |\norm{a_{w_k} \cdots a_{w_1} a_y \psi_{N,t}} \norm{a_{u_k} \cdots a_{u_1} a_y \psi_{N,t}}\\
	=: & I_4 + \mathit{II}_4 
	\end{align*}
	Focusing on $I_4$, we split the integral as follows
	\begin{align*}
	I_4 & =  (2\pi)^3 \hbar^{3-\frac{3}{2}k} \sum_{j=1}^k   \dotsint (\mathrm{d}q\mathrm{d}w\mathrm{d}u)^{\otimes k} \bigg|\prod_{n=1}^k  \left(\rchi_{(w_n-u_n)\in ({\Omega_\hbar^\alpha})^c}+\rchi_{(w_n-u_n)\in {\Omega_\hbar^\alpha}} \right) \dotsint (\mathrm{d}p)^{\otimes k}  \nabla_{p_j}\Phi(q_1, p_1,\dots,q_k,p_k)\\
	&\quad\cdot  e^{\frac{\mathrm{i}}{\hbar} \sum_{m=1}^k p_m \cdot (w_m-u_m)} \bigg| \iint \mathrm{d}y\mathrm{d}\widetilde{q}_{k+1} \int_0^1 \mathrm{d}s \left|su_j + (1-s)w_j-q_j\right| 
	\left|f \left(\frac{w_n-q_n}{\sqrt{\hbar}}\right)f \left(\frac{u_n-q_n}{\sqrt{\hbar}}\right)\right| |f(\widetilde{q}_{k+1})|^2\\
	&\quad\cdot \rchi_{|w_n-q_n|\leq\sqrt{\hbar}R}\rchi_{|u_n-q_n|\leq\sqrt{\hbar}R} |\norm{a_{w_k} \cdots a_{w_1} a_y \psi_{N,t}} \norm{a_{u_k} \cdots a_{u_1} a_y \psi_{N,t}}.
	\end{align*}
	where ${\Omega_\hbar^\alpha}$ is defined as in \eqref{estimate_oscillation_omega}. We do similar computations for $\mathit{II}_4$,
	\begin{align*}
	\mathit{II}_4  = & (2\pi)^3 \hbar^{3-\frac{3}{2}k} \sum_{j=1}^k    \dotsint (\mathrm{d}q\mathrm{d}w\mathrm{d}u)^{\otimes k} \bigg|\prod_{n=1}^k \dotsint (\mathrm{d}p)^{\otimes k} \left(\rchi_{(w_n-u_n)\in ({\Omega_\hbar^\alpha})^c}+\rchi_{(w_n-u_n)\in {\Omega_\hbar^\alpha}} \right) \nabla_{p_j}\Phi(q_1, p_1,\dots,q_k,p_k) \\
	&\quad\cdot e^{\frac{\mathrm{i}}{\hbar} p_n \cdot (w_n-u_n)} \bigg| \iint \mathrm{d}y\mathrm{d}\widetilde{q}_{k+1} \left|\sqrt{\hbar}\widetilde{q}_{k+1} \right|
	\left|f \left(\frac{w_n-q_n}{\sqrt{\hbar}}\right)f \left(\frac{u_n-q_n}{\sqrt{\hbar}}\right)\right| |f(\widetilde{q}_{k+1})|^2\\
	&\quad\cdot \rchi_{|w_n-q_n|\leq\sqrt{\hbar}R}\rchi_{|u_n-q_n|\leq\sqrt{\hbar}R} |\norm{a_{w_k} \cdots a_{w_1} a_y \psi_{N,t}} \norm{a_{u_k} \cdots a_{u_1} a_y \psi_{N,t}}.
	\end{align*}
	Repeating the proof of Proposition \ref{lemma_vla_k1_convg}, namely in \eqref{proof_I_3} and \eqref{proof_II_3}, as well as the proof for estimate \eqref{mk_BBFKY_to_infty_1}, we eventually obtain the desired estimates.
\end{proof}

\vspace*{\fill}

\textit{Acknowledgements:}
We are also grateful to the anonymous referee for carefully reading our manuscript and providing helpful comments.
We acknowledge support by the Deutsche Forschungsgemeinschaft through the grants CH 955/4-1.
Jinyeop Lee was partially supported by Samsung
Science and Technology Foundation (SSTF-BA1401-51) and by the National
Research Foundation of Korea(NRF) grant funded by the Korea government(MSIT)
(NRF-2019R1A5A1028324 and NRF-2020R1F1A1A01070580).

\end{document}